\documentclass[proceedings]{stacs}
\stacsheading{2010}{359-370}{Nancy, France}
\firstpageno{359}

\usepackage{epsfig}
\usepackage{multicol}
\usepackage{pstricks,pst-grad}
\usepackage{amssymb}
\usepackage{amsthm}
\usepackage{amsmath}
\usepackage{color}
\usepackage[ruled, vlined]{algorithm2e}
\usepackage{float}
\usepackage{wrapfig}
\usepackage{wasysym}
\usepackage{tikz}
\usetikzlibrary{arrows,shapes}
\usepackage{multirow}
\usepackage{subfigure}
\usepackage{floatflt}
\usepackage{float}
\usepackage{times}

\usepackage{ifthen}
\newcommand{\techRep}{false} 
\newcommand{\iftechrep}{\ifthenelse{\equal{\techRep}{true}}}

\theoremstyle{plain}
\theoremstyle{definition}

\newcommand{\vd}{\mathbf{d}}
\newcommand{\vp}{\mathbf{p}}
\newcommand{\vt}{\mathbf{t}}
\newcommand{\vu}{\mathbf{u}}
\newcommand{\vv}{\mathbf{v}}
\newcommand{\vx}{\mathbf{x}}
\newcommand{\vy}{\mathbf{y}}
\newcommand{\vz}{\mathbf{z}}
\newcommand{\vDelta}{\mathbf{\Delta}}
\newcommand{\ub}{\mathbf{ub}}
\newcommand{\lb}{\mathbf{lb}}
\newcommand{\vzero}{\overline{0}}
\newcommand{\vone}{\overline{1}}
\newcommand{\vtwo}{\overline{2}}
\newcommand{\vX}{\overline{X}}
\newcommand{\ps}[1]{\mathbf{p}^{(#1)}}
\newcommand{\xs}[1]{\mathbf{x}^{(#1)}}

\newcommand{\lbs}[1]{\lb^{(#1)}}
\newcommand{\ubs}[1]{\ub^{(#1)}}
\newcommand{\hs}[1]{h^{(#1)}}
\newcommand{\ks}[1]{\mathbf{\boldsymbol\kappa}^{(#1)}}
\newcommand{\ns}[1]{\mathbf{\boldsymbol\nu}^{(#1)}}

\newcommand{\Id}{\mathit{Id}}
\newcommand{\R}{\mathbb{R}}
\newcommand{\Ne}{\mathcal{N}}
\newcommand{\Rat}{\mathbb{Q}}

\newcommand{\Float}[2] {#1\hspace{1mm}\rotatebox[x=2mm,y=1mm]{180}{$\rightsquigarrow$}\,#2}

\definecolor{orange3}{rgb}{1.0,0.2538,0.1681}
\definecolor{blau}{rgb}{0,.39608,0.74118}
\definecolor{rot}{rgb}{0.79216,.12941,0.24706}
\definecolor{dgruen}{rgb}{0,.7,0}

\newenvironment{qtheorem}[2][]{%
{${}$\\[1mm]
\noindent \bf Theorem #2#1.}
\begin{itshape}%
}{%
\end{itshape}%
}

\newenvironment{qproposition}[2][]{%
{${}$\\[1mm]
\noindent\bf Proposition #2#1.}
\begin{itshape}%
}{%
\end{itshape}%
}

\begin{document}
\title[Computing Least Fixed Points of Probabilistic Systems of Polynomials]{Computing Least Fixed Points of \\ Probabilistic Systems of Polynomials}

\author{J.~Esparza}{Javier Esparza}
\author{A.~Gaiser}{Andreas Gaiser}
\author{S.~Kiefer}{Stefan Kiefer}

\address[TUM]{Fakult{\"a}t f{\"u}r Informatik, Technische Universit{\"a}t M{\"u}nchen, Germany}  
\email{{esparza,gaiser,kiefer}@model.in.tum.de}  

\begin{abstract}
We study systems of equations of the form
 $X_1 = f_1(X_1, \ldots, X_n), \ldots, X_n = f_n(X_1, \ldots, X_n)$ where each $f_i$ is a polynomial with nonnegative coefficients that add up to~$1$.
The least nonnegative solution, say~$\mu$, of such equation systems is central to problems from various areas,
 like physics, biology, computational linguistics and probabilistic program verification.
We give a simple and strongly polynomial algorithm to decide whether $\mu=(1,\ldots,1)$ holds.
Furthermore, we present an algorithm that computes reliable sequences of lower and upper bounds on~$\mu$, converging linearly to~$\mu$.
Our algorithm has these features despite using inexact arithmetic for efficiency.
We report on experiments that show the performance of our algorithms.
\end{abstract}

\maketitle

\section{Introduction}

We study how to efficiently compute the least nonnegative solution of
an equation system of the form
$$\begin{array}{ccc}
X_1 = f_1(X_1, \ldots, X_n) & \ldots & X_n = f_n(X_1, \ldots, X_n)\;,
\end{array}$$
where, for every $i \in \{1,\ldots,n\}$,
$f_i$ is a polynomial over $X_1, \ldots, X_n$ with positive rational
coefficients that \emph{add up to~1}.\footnote{Later, we allow that the coefficients add up to {\em at most~$1$}.}
The solutions are the fixed points of the function
$f: \mathbb{R}^n \rightarrow \mathbb{R}^n$ with $f = (f_1, \ldots, f_n)$.
We call $f$ a {\em probabilistic system of polynomials} (short: {\em PSP}).
E.g., the PSP
$$f(X_1, X_2) = \left( \, \frac{1}{2}X_1 X_2 + \frac{1}{2} \; , \; \frac{1}{4}X_2 X_2 + \frac{1}{4} X_1 + \frac{1}{2} \,\right)$$
induces the equation system
$$\textstyle X_1 = \frac{1}{2}X_1 X_2 + \frac{1}{2} \qquad X_2 = \frac{1}{4}X_2 X_2 + \frac{1}{4} X_1 + \frac{1}{2} \; .$$
Obviously, $\overline{1} = (1, \ldots, 1)$ is a fixed point of every PSP.
By Kleene's theorem, every PSP has a least nonnegative fixed point
(called just least fixed point in what follows), given by the limit
of the sequence $\overline{0}, f(\overline{0}), f(f(\overline{0})), \ldots$

PSPs are important in different areas of the theory of stochastic processes
and computational models. A fundamental result of the theory of
branching processes, with numerous applications in physics, chemistry and biology
(see e.g.\ \cite{harris,athreya}), states that extinction probabilities of species are equal to the least fixed point of a PSP.
The same result has been recently shown for the probability of termination of certain probabilistic recursive programs \cite{etessamiyannakakis,EKM04}.
The consistency of stochastic context-free grammars, a problem of interest in statistical
natural language processing, also reduces to checking whether the least fixed point of a PSP equals~$\overline{1}$
(see e.g.~\cite{foundations_nlp}).

Given a PSP~$f$ with least fixed point $\mu_f$, we study how to efficiently solve the following two problems:
(1) decide whether $\mu_f = \overline{1}$, and
(2) given a rational number $\epsilon > 0$, compute $\lb, \ub \in \Rat^n$ such that $\lb \leq \mu_f \leq \ub$ and
$\ub - \lb \leq \overline{\epsilon}$ (where $\vu \le \vv$ for vectors~$\vu, \vv$ means $\le$ in all components).
While the motivation for Problem~(2) is clear (compute the probability of extinction with a given accuracy),
 the motivation for Problem~(1) requires perhaps some explanation.
In the case study of Section~\ref{sub:case-neutron} we consider a family of PSPs, taken from~\cite{harris}, modelling
 the neutron branching  process in a ball of radioactive material of radius~$D$ (the family is parameterized by~$D$).
The least fixed point is the probability that a neutron produced through spontaneous fission {\em does not}
 generate an infinite ``progeny'' through successive collisions with atoms of the ball;
loosely speaking, this is the probability that the neutron {\em does not} generate a chain reaction and the ball {\em does not} explode.
Since the number of atoms in the ball is very large, spontaneous fission produces many neutrons per second, 
and so even if the probability that a given neutron produces a chain reaction is very small, the ball will explode with large
probability in a very short time.
It is therefore important to determine the largest radius~$D$ at which the probability of no chain reaction is still~$1$
(usually called the {\it critical radius}).
An algorithm for Problem~(1) allows to compute the critical radius using binary search.
A similar situation appears in the analysis of parameterized probabilistic programs.
In~\cite{etessamiyannakakis,EKM04} it is shown that the question whether a probabilistic program
almost surely terminates can be reduced to Problem~(1).
Using binary search one can find the ``critical'' value of the parameter for which
the program may not terminate any more.

Etessami and Yannakakis show in~\cite{etessamiyannakakis} that
Problem~(1) can be solved in polynomial time by a reduction to (exact) Linear Programming (LP), which is not known to be strongly polynomial.
Our first result reduces Problem~(1) to solving a system of linear equations, resulting in a strongly polynomial algorithm for Problem~(1). The Maple library offers exact arithmetic solvers for LP and systems
of linear equations, which we use to test the performance of our new algorithm.
In the neutron  branching process discussed above we obtain speed-ups of about one order of magnitude
with respect to LP.

The second result of the paper is, to the best of our knowledge, the first practical algorithm for Problem~(2).
Lower bounds for $\mu_f$ can be computed using Newton's method for approximating a root of the function $f(\overline{X}) - \overline{X}$.
This has recently been investigated in detail~\cite{etessamiyannakakis,EKL07:stoc,EKL08:stacs}.
However, Newton's method faces considerable numerical problems.
Experiments show that naive use of exact arithmetic is inefficient, while floating-point computation
leads to false results even for very small systems. For instance, the PReMo tool~\cite{premo},
which implements Newton's method with floating-point arithmetic for efficiency, reports $\mu_f \ge \vone $ for a
PSP with only 7 variables and small coefficients, although $\mu_f < \vone$
is the case (see Section~\ref{sub:case-h-n}).

Our algorithm produces a sequence of guaranteed lower and upper bounds, both of which converge
linearly to $\mu_f$. Linear convergence means that, loosely speaking,
the number of accurate bits of
the bound is a linear function of the position of the bound in the sequence.
The algorithm is based on the following idea.
Newton's method is an iterative procedure that, given a current lower bound~$\lb$ on~$\mu_f$, applies
a certain operator~$\Ne$ to it, yielding a new, more precise lower bound~$\Ne(\lb)$.
Instead of computing $\Ne(\lb)$ using exact arithmetic, our algorithm computes {\em two} consecutive
Newton steps, i.e., $\Ne(\Ne(\lb))$, using {\em inexact} arithmetic.
Then it checks if the result satisfies a carefully chosen condition.
If so, the result is taken as the next lower bound.
If not, then the precision is increased, and the computation redone.
The condition is eventually satisfied, assuming the results of computing with increased precision
converge to the exact result.
Usually, the repeated inexact computation is much faster than the exact one.
At the same time, a careful (and rather delicate) analysis shows that the sequence of lower
bounds converges linearly to~$\mu_f$.

Computing {\em upper} bounds is harder, and seemingly has not been considered in the literature before.
Similarly to the case of lower bounds, we apply $f$ twice to~$\ub$, i.e.,
we compute $f(f({\bf ub}))$ with increasing precision until a condition holds.
The sequence so obtained may not even converge to~$\mu_f$.
So we need to introduce a further operation, after which we can then prove linear convergence.

We test our algorithm on the neutron branching process. The time
needed to obtain lower and upper bounds on the probability of no
explosion with $\epsilon = 0.0001$ lies below the time needed to check,
using exact LP, whether this probability is $1$ or smaller than one. That
is, in this case study our algorithm is faster, and provides more information.

The rest of the paper is structured as follows.
We give preliminary definitions and facts in Section~\ref{sec:prelim}.
Sections \ref{sec:consistency} and~\ref{sec:inex} present our algorithms for solving Problems (1) and~(2),
and report on their performance on some case studies. Section~\ref{sec:conclusions} contains our conclusions.
\iftechrep{A shorter version of this paper will appear in \emph{27th International Symposium on Theoretical Aspects of Computer Science (STACS 2010)}.}{The full version of the paper, including all proofs, can be found in~\cite{EGK10:stacsTechRep}.} 
\section{Preliminaries} \label{sec:prelim}
\paragraph{\em Vectors and matrices.}
We use bold letters for designating (column) vectors, e.g.\ $\mathbf{v} \in \mathbb{R}^n$.
We write $\overline{s}$ with $s \in \mathbb{R}$ for the vector $(s, \ldots, s)^\top \in \mathbb{R}^n$ (where $^\top$ indicates transpose),
 if the dimension $n$ is clear from the context.
The $i$-th component of $\mathbf{v} \in \mathbb{R}^n$ will be denoted by~$\mathbf{v}_i$.
We write $\vx = \vy$ (resp.\ $\vx \le \vy$ resp.\ $\vx \prec \vy$) if $\vx_i = \vy_i$ (resp.\ $\vx_i \le \vy_i$ resp.\ $\vx_i < \vy_i$)
 holds for all $i \in \{1,\ldots,n\}$.
By $\vx < \vy$ we mean $\vx \le \vy$ and $\vx \ne \vy$.
\iftechrep{\par}{}
By~$\mathbb{R}^{m \times n}$ we denote the set of real matrices with $m$ rows and $n$ columns.
We write $\Id$ for the identity matrix.
For a square matrix~$A$, we denote by $\rho(A)$ the \emph{spectral radius} of~$A$, i.e., the maximum of the absolute values of the eigenvalues.
A matrix is {\em nonnegative} if all its entries are nonnegative.
A nonnegative matrix $A\in \mathbb{R}^{n \times n}$ is \emph{irreducible} if for every $k,l \in \{1, \ldots, n\}$
 there exists an  $i \in \mathbb{N}$ so that $(A^i)_{kl} \not = 0$.

\paragraph{\em Probabilistic Systems of Polynomials.}
We investigate equation systems of the form
$$\begin{array}{ccc}
X_1 = f_1(X_1, \ldots, X_n) & \ldots & X_n = f_n(X_1, \ldots, X_n),
\end{array}$$
where the $f_i$ are polynomials in the variables $X_1, \ldots, X_n$ with positive real coefficients,
 and for every polynomial~$f_i$ the sum of its coefficients is \emph{at most}~$1$.
The vector $f := (f_1,\ldots,f_n)^\top$ is called a \emph{probabilistic system of polynomials} (PSP for short)
 and is identified with its induced function $f : \mathbb{R}^n \rightarrow \mathbb{R}^n$.
If $X_1, \ldots, X_n$ are the formal variables of~$f$, we define $\vX := (X_1, \ldots, X_n)^\top$ and $\text{Var}(f) := \{X_1, \ldots, X_n\}$.
We assume that $f$ is represented as a list of polynomials, and each polynomial is a list of its monomials.
If $S \subseteq \{X_1,\ldots,X_n\}$, then $f_S$ denotes the result of removing the polynomial
$f_i(X_1, \ldots, X_n)$ from $f$ for every $x_i \notin S$; further, given $\vx \in \R^n$ and $B \in \R^{n\times n}$, we denote by $\vx_S$ and $B_{SS}$ the vector and the matrix obtained from~$\vx$ and~$B$ by removing the entries with indices~$i$ such that $X_i \not\in S$.
The coefficients are represented as fractions of positive integers.
The {\em size} of~$f$ is the size of that representation.
The \emph{degree} of~$f$ is the maximum of the degrees of $f_1, \ldots, f_n$.
 PSPs of degree $0$ (resp.~$1$ resp.~$\mathord{>}1$) are called {\em constant} (resp. \emph{linear} resp.\ {\em superlinear}).
PSPs~$f$ where the degree of each $f_i$ is at least~$2$ are called {\em purely superlinear}.
We write $f'$ for the \emph{Jacobian} of~$f$, i.e., the matrix of first partial derivatives of~$f$.

Given a PSP $f$, a variable $X_i$ {\em depends directly} on a 
variable~$X_j$ if $X_j$ ``occurs'' in~$f_i$,
 more formally if $\frac{\partial f_i}{\partial X_j}$ is not the constant~$0$.
A variable~$X_i$ {\em depends} on~$X_j$ if $X_i$ depends directly on~$X_j$ or there is a variable~$X_k$ such that
 $X_i$ depends directly on~$X_k$ and $X_k$ depends on~$X_j$.
We often consider the {\em strongly connected components} (or SCCs for short) of the dependence relation.
The SCCs of a PSP can be computed in linear time using e.g. Tarjan's algorithm. 
An SCC~$S$ of a PSP~$f$ is {\em constant} resp.\ {\em linear} resp.\ {\em superlinear} resp.\ {\em purely superlinear}
 if the PSP~$\tilde f$ has the respective property, 
 where $\tilde f$ is obtained by restricting~$f$ to the $S$-components and replacing all variables not in~$S$ by the constant~$1$.
A PSP is an {\em scPSP} if it is not constant and consists of only one SCC.
Notice that a PSP~$f$ is an scPSP if and only if $f'(\vone)$ is irreducible.

A fixed point of a PSP~$f$ is a vector~$\vx \ge \vzero$ with $f(\vx) = \vx$.
By Kleene's theorem, there exists a least fixed point~$\mu_f$ of~$f$, i.e., $\mu_f \le \vx$ holds for every fixed point~$\vx$.
Moreover, the sequence $\vzero, f(\vzero), f(f(\vzero)), \ldots$ converges to~$\mu_f$.
Vectors~$\vx$ with $\vx \le f(\vx)$ (resp.\ $\vx \ge f(\vx)$) are called {\em pre-fixed} (resp.\ {\em post-fixed}) points.
Notice that the vector~$\vone$ is always a post-fixed point of a PSP~$f$, due to our assumption on the coefficients of a PSP.
By Knaster-Tarski's theorem, $\mu_f$ is the least post-fixed point, so we always have $\vzero \le \mu_f \le \vone$.
It is easy to detect and remove all components~$i$ with $(\mu_f)_i = 0$ by a simple round-robin method (see e.g.~\cite{EKL08:stacs}),
 which needs linear time in the size of~$f$.
We therefore assume in the following that $\mu_f \succ \vzero$.

\section{An algorithm for consistency of PSPs} \label{sec:consistency}

Recall that for applications like the neutron branching process
it is crucial to know exactly whether $\mu_f = \overline{1}$ holds.
We say a PSP~$f$ is {\em consistent} if $\mu_f = \overline{1}$; otherwise it is \emph{inconsistent}.
Similarly, we call a component~$i$ consistent if $(\mu_f)_i = 1$.
We present a new algorithm for the consistency problem, i.e., the problem to check a
PSP for consistency.

It was proved in~\cite{etessamiyannakakis} that consistency is checkable in polynomial time
by reduction to Linear Programming (LP). We first observe that consistency of general
PSPs can be reduced to consistency of scPSPs by computing the DAG of SCCs, and checking consistency
SCC-wise~\cite{etessamiyannakakis}:  Take any bottom SCC~$S$, and check
the consistency of $f_S$. (Notice that $f_S$ is either
constant or an scPSP; if constant, $f_S$ is consistent if{}f $f_S=1$, if an scPSP,
we can check its consistency by assumption.) If $f_S$ is inconsistent, then so is $f$, and we are done.
If $f_S$ is consistent, then we remove every $f_i$ from~$f$ such that $x_i \in S$,
replace all variables of~$S$ in the remaining polynomials by
the constant~$1$, and iterate (choose a new bottom SCC, etc.). Note that this algorithm
processes each polynomial at most once, as every variable belongs to exactly one SCC.

It remains to reduce the consistency problem for scPSPs to LP. The first step
is:
\begin{proposition}\cite{harris,etessamiyannakakis} \; \label{prop:spektralradius}
 An scPSP~$f$ is consistent if{}f $\rho(f'(\overline{1})) \leq 1$ \
(i.e., if{}f the spectral radius of the Jacobi matrix $f'$ evaluated at the vector $\vone$
is at most $1$).
\end{proposition}

\noindent The second step consists of observing that the
 matrix $f'(\overline{1})$ of an scPSP $f$ is irreducible and nonnegative.
It is shown in~\cite{etessamiyannakakis} that $\rho(A) \leq 1$
holds for an irreducible and nonnegative matrix $A$ if{}f the system of inequalities
\begin{equation} \label{eq:lp_spectralradius}
 A \mathbf{x} \geq \vx + \vone \text{ , } \mathbf{x} \geq \overline{0}
\end{equation}

\noindent is infeasible. However, no strongly polynomial algorithm for LP
is known, and we are not aware that (\ref{eq:lp_spectralradius}) falls
within any subclass solvable in strongly polynomial
time~\cite{geometricalgorithms}.

We provide a very simple, strongly polynomial time algorithm to check
whether $\rho(f'(\overline{1})) \leq 1$ holds. We need some results from Perron-Frobenius theory (see e.g. \cite{matrixbuch}).
\begin{lemma} \label{lemma:pf1}
Let $A \in \mathbb{R}^{n \times n}$ be nonnegative and irreducible.
\begin{itemize}
\item[(1)] $\rho(A)$ is a \emph{simple} eigenvalue of $A$.
\item[(2)] There exists an eigenvector $\mathbf{v} \succ \vzero$ with $\rho(A)$ as eigenvalue.
\item[(3)] Every eigenvector $\mathbf{v} \succ \vzero$ has $\rho(A)$ as eigenvalue.
\item[(4)] For all $\alpha, \beta \in \mathbb{R} \setminus \{0\}$ and $\mathbf{v} > \overline{0}$:
if $\alpha \mathbf{v} < A \mathbf{v} < \beta \mathbf{v}$,
then $\alpha < \rho(A) < \beta$.
\end{itemize}
\end{lemma}
\noindent The following lemma is the key to the algorithm:
\newcommand{\stmtlemmatrixspectralradius}{
 Let $A\in\R^{n \times n}$ be nonnegative and irreducible.
 \begin{enumerate}
   \item[(a)]
    Assume there is $\vv \in \R^n \setminus \{\vzero\}$
    such that $(\Id - A) \vv = \vzero$. Then $\rho(A) \le 1$ if{}f $\vv \succ \vzero$ or $\vv \prec \vzero$.
   \item[(b)]
    Assume $\vv = \vzero$ is the only solution of $(\Id - A) \vv = \vzero$.
    Then there exists a unique $\vx \in \R^n$ such that
    $(\Id - A) \vx = \vone$, and $\rho(A) \le 1$ if{}f $\vx \ge \vone$ and $A \vx < \vx$.
 \end{enumerate}
}
\begin{lemma} \label{lem:matrix-spectral-radius}
 \stmtlemmatrixspectralradius
\end{lemma}
\proof \mbox{}
\begin{enumerate}
 \item[(a)]
   From $(\Id - A) \vv = \overline{0}$ it follows $A \vv =\vv$.
   We see that $\vv$ is an eigenvector of~$A$ with eigenvalue~$1$.
   So $\rho(A) \geq 1$.

   \noindent($\Leftarrow$): As both $\vv$ and $-\vv$ are eigenvectors of~$A$ with eigenvalue~$1$, we can assume w.l.o.g.\ that $\vv \succ \vzero$.
   By Lemma~\ref{lemma:pf1}(3), $\rho(A)$ is the eigenvalue
   of $\vv$, and so $\rho(A)= 1$.

    \noindent($\Rightarrow$): Since $\rho(A) \leq 1$ and $\rho(A) \geq 1$, it follows that $\rho(A) = 1$. By Lemma~\ref{lemma:pf1}(1) and (2), the eigenspace of the eigenvalue~$1$ is one-dimensional and contains a vector $\vx \succ \vzero$. So $\vv = \alpha \cdot \vx$  for some
     $\alpha \in \mathbb{R}, \alpha \not = 0$. If $\alpha > 0$, we have $\vv \succ \vzero$, otherwise $\vv \prec \vzero$.
 \item[(b)]
  With the assumption and basic facts from linear algebra it follows that $(Id - A)$ has full rank and therefore
  $(\Id - A) \vx = \vone$ has a unique solution $\vx$. We still have to prove the second part of the conjunction:

   \noindent($\Leftarrow$):  Follows directly from Lemma~\ref{lemma:pf1}(4).

   \noindent($\Rightarrow$): Let $\rho(A) \le 1$.
     Assume for a contradiction that $\rho(A) = 1$.
     Then, by Lemma~\ref{lemma:pf1}(1), the matrix~$A$ would have an eigenvector~$\vv \neq \vzero$ with eigenvalue~$1$, so $(\Id - A) \vv = \vzero$, contradicting the assumption.
     So we have, in fact, $\rho(A) < 1$.
     By standard matrix facts (see e.g.~\cite{matrixbuch}), this implies that $(\Id - A)^{-1} = A^* = \sum_{i=0}^\infty A^i$ exists,
      and so we have $\vx = (\Id - A)^{-1} \vone = A^* \vone \ge \vone$.
     Furthermore,  $A \vx = \sum_{i=1}^\infty A^i \vone < \sum_{i=0}^\infty A^i \vone = \vx$. \qed
\end{enumerate}

In order to check whether $\rho(A) \leq 1$, we first solve the system $(\Id - A) \vv = \vzero$
using Gaussian elimination. If we find a vector
$\vv \not= \vzero$ such that $(Id - A) \vv = \vzero$, we apply Lemma~\ref{lem:matrix-spectral-radius}(a).
If $\vv = \vzero$ is the only solution of $(Id - A) \vv = \vzero$, we solve
$(\Id - A) \vv = \vone$ using Gaussian elimination again,
and apply Lemma~\ref{lem:matrix-spectral-radius}(b).
Since Gaussian elimination of a rational $n$-dimensional linear equation system
can be carried out in strongly polynomial time using $O(n^3)$ arithmetic operations
(see e.g. ~\cite{geometricalgorithms}), we obtain:

\begin{proposition} \label{prop:decide-spectral-radius}
Given a nonnegative irreducible matrix~$A\in\R^{n \times n}$, one can
decide in strongly polynomial time, using $O(n^3)$ arithmetic operations,
whether $\rho(A) \le 1$.
\end{proposition}

Combining Propositions \ref{prop:spektralradius} and~\ref{prop:decide-spectral-radius} directly
yields an algorithm for checking the consistency of scPSPs.
Extending it to multiple SCCs as above, we get:

\begin{theorem} \label{thm:consistency}
  Let $f(X_1, \ldots, X_n)$ be a PSP.
  There is a strongly polynomial time algorithm that uses $O(n^3)$ arithmetic operations and
  determines the consistency of $f$.
\end{theorem}

\subsection{Case study: A family of ``almost consistent'' PSPs} \label{sub:case-h-n}

In this section, we illustrate some issues faced by algorithms that solve the consistency problem.
Consider the following family $\hs{n}$ of scPSPs, $n \geq 2$:
\iftechrep{
\begin{equation*}
  \hs{n} = \left( \; 0.5 X_1^2 + 0.1 X_n^2 + 0.4 \;,\;  0.01X_1^2 + 0.5 X_2 + 0.49 \; , \; \ldots \; , 0.01X_{n-1}^2 + 0.5X_n + 0.49 \; \right)^\top\;.
 \end{equation*}
}
{\begin{equation*}
  \hs{n} = \left( \; 0.5 X_1^2 + 0.1 X_n^2 + 0.4 \;,\;  0.01X_1^2 + 0.5 X_2 + 0.49 \; , \; \ldots \; , 0.01X_{n-1}^2 + 0.5X_n + 0.49 \; \right)^\top\;.
\end{equation*}
}
It is not hard to show that $\hs{n}(\vp) \prec \vp$ holds for $\vp = (1 - 0.02^n, \ldots, 1 - 0.02^{2n-1})^\top$,
so we have $\mu_{\hs{n}} \prec \vone$ by Proposition~\ref{prop:prefix-postfix}, i.e., the $\hs{n}$ are inconsistent.

The tool PReMo \cite{premo} relies on Java's floating-point arithmetic to compute
approximations of the least fixed point of a PSP.
We invoked PReMo for computing approximants of~$\mu_{\hs{n}}$ for different values of~$n$
between $5$ and~$100$.
Due to its fixed precision, PReMo's approximations for $\mu_{\hs{n}}$ are $\ge 1$
in all components if $n \geq 7$.
This might lead to the wrong conclusion that $\hs{n}$ is consistent.

Recall that the consistency problem can be solved by checking the feasibility of the
system~\eqref{eq:lp_spectralradius} with $A = f'(\overline{1})$.
We checked it with lp\_solve, a well-known LP tool using hardware floating-point arithmetic.
The tool wrongly states that \eqref{eq:lp_spectralradius}
has no solution for $\hs{n}$-systems with $n > 10$.
This is due to the fact that the solutions cannot be represented adequately using machine number precision.%
\footnote{%
The mentioned problems of PReMo and lp\_solve are not due to the fact that the coefficients of~$\hs{n}$ cannot be properly represented using 
basis~2: The problems persist if one replaces the coefficients of~$\hs{n}$ by similar numbers exactly representable by machine numbers.}
Finally, we also checked feasibility with Maple's Simplex package, which uses exact arithmetic,
and compared its performance with the implementation, also in Maple, of our consistency algorithm.
Table~\ref{tab:runtime-hn} shows the results. 
Our algorithm clearly outperforms the LP approach.
For more experiments see Section~\ref{sub:case-neutron}.

\begin{table}
\begin{tabular}[tbp]{l|r|r|r|r|r|r}
                                            & $n=25$   & $n=100$     & $n=200$ & $n=400$ & $n=600$  & $n=1000$\\ \hline
Exact LP                                    & $<1$ sec & 2 sec       & 8 sec   & 67 sec  & 208 sec  & $>$ 2h \\ \hline
Our algorithm                               & $<1$ sec & $<1$ sec    & 1 sec   & 4 sec   & 10  sec  & 29 sec \\ 
\end{tabular}
\caption{Consistency checks for $\hs{n}$-systems: Runtimes of different approaches.}
\label{tab:runtime-hn}
\end{table}

\section{Approximating~$\mu_f$ with inexact arithmetic} \label{sec:inex}

It is shown in~\cite{etessamiyannakakis} that $\mu_f$ may not be representable by roots, so one can only approximate~$\mu_f$.
In this section we present an algorithm that computes two sequences, $(\lbs{i})_i$ and $(\ubs{i})_i$, such that
 $\lbs{i} \le \mu_f \le \ubs{i}$ and $\lim_{i\to\infty} \ubs{i} - \lbs{i} = \vzero$.
In words: $\lbs{i}$ and~$\ubs{i}$ are lower and upper bounds on~$\mu_f$, respectively, and the sequences converge to~$\mu_f$.
Moreover, they converge linearly, meaning that the {\em number of accurate bits} of~$\lbs{i}$ and~$\ubs{i}$ are linear functions of~$i$.
(The number of accurate bits of a vector~$\vx$ is defined as the greatest number~$k$
 such that $|(\mu_f - \vx)_j|/|(\mu_f)_j| \le 2^{-k}$ holds for all $j \in \{1,\ldots,n\}$.)
These properties are guaranteed even though our algorithm uses inexact
arithmetic: Our algorithm detects numerical problems due to rounding errors,
recovers from them, and increases the precision of the arithmetic as needed.
Increasing the precision dynamically is, e.g., supported by the
GMP library~\cite{GMP}.

\newcommand{\gs}[1]{g^{(#1)}}
Let us make precise what we mean by increasing the precision.
Consider an elementary operation~$g$, like multiplication, subtraction,
etc., that operates on two input numbers $x$ and~$y$.
We can {\em compute $g(x,y)$ with increasing precision} if there is
a procedure that on input $x, y$ outputs a sequence
$\gs{1}(x,y), \gs{2}(x,y), \ldots$ that converges to~$g(x,y)$.
Note that there are no requirements on the convergence speed of this procedure ---
 in particular, we do not require that there is an~$i$ with $\gs{i}(x,y) = g(x,y)$.
This procedure, which we assume exists, allows to implement
{\em floating assignments} of the form
 \[
  \Float{z}{g(x,y)} \textbf{ such that } \phi(z)
 \]
with the following semantics: $z$ is assigned the value $\gs{i}(x,y)$,
where $i \geq 1$ is the smallest index such that $\phi(\gs{i}(x,y))$ holds.
We say that the assignment is {\em valid} if $\phi(g(x,y))$ holds and
$\phi$ involves
only continuous functions and strict inequalities.
Our assumption on the arithmetic guarantees that
(the computation underlying) a valid floating assignment terminates.
As ``syntactic sugar'', more complex operations
(e.g., linear equation solving) are also allowed in floating assignments,
because they can be decomposed into elementary operations.

We feel that any implementation of arbitrary precision arithmetic should
satisfy our requirement that the computed values converge to the exact result.
For instance, the documentation of the GMP library~\cite{GMP} states:
 ``Each function is defined to calculate with `infinite precision' followed by a truncation to the destination precision,
 but of course the work done is only what's needed to determine a result under that definition.''

To approximate the least fixed point of a PSP, we first transform it into a certain normal form.
A purely superlinear PSP~$f$ is called {\em perfectly superlinear}
 if every variable depends directly on itself and every superlinear SCC is purely superlinear.
The following proposition states that any PSP~$f$ can be made perfectly superlinear.

\newcommand{\propnormalform}{
 Let $f$ be a PSP of size~$s$.
 We can compute in time $O(n \cdot s)$ a perfectly superlinear PSP~$\tilde f$ with $\text{Var}(\tilde f) = \text{Var}(f) \cup \{\tilde X\}$
  of size $O(n \cdot s)$ such that $\mu_f = ( \mu_{\tilde f} )_{\text{Var}(f)}$.
}
\begin{proposition} \label{prop:normal-form}
 \propnormalform
\end{proposition}

\subsection{The algorithm}

The algorithm receives as input a perfectly superlinear PSP $f$ and an error bound $\epsilon > 0$,
and returns vectors $\lb, \ub$ such that $\lb \le \mu_f \le \ub$ and $\ub - \lb \le \overline{\epsilon}$.
A first initialization step requires to compute a vector~$\vx$
with \mbox{$\vzero \prec \vx \prec f(\vx)$}, i.e., a ``strict'' pre-fixed point. This is done in Section~\ref{sub:pre-fixed}.
The algorithm itself is described in Section~\ref{sub:upper-and-lower}.

\subsubsection{Computing  a strict pre-fixed point} \label{sub:pre-fixed}

Algorithm~\ref{alg:compute-prefix} computes a strict pre-fixed point:

\vspace{2mm}
\begin{algorithm}[H]
\KwIn{perfectly superlinear PSP $f$}
\KwOut{$\vx$ with $\vzero \prec \vx \prec f(\vx) \prec \vone$}
 $\vx \leftarrow \vzero$\;
 \While{$\vzero \not\prec \vx$}{
  $Z \leftarrow \{i \mid 1 \le i \le n, f_i(\vx) = 0\}$\;
  $P \leftarrow \{i \mid 1 \le i \le n, f_i(\vx) > 0\}$\;
  $\vy_Z \leftarrow \overline{0}$\;
  $\Float{\vy_P}{f_P(\vx)}$ \textbf{such that} $\vzero \prec \vy_P \prec f_P(\vy) \prec \vone$\;
  $\vx \leftarrow \vy$\;
 }
\label{alg:compute-prefix}
\caption{Procedure \texttt{computeStrictPrefix}}
\end{algorithm}

\newcommand{\propalgcomputeprefix}{
  Algorithm~\ref{alg:compute-prefix} is correct and terminates after at most~$n$ iterations.
}
\begin{proposition} \label{prop:alg-compute-prefix}
 \propalgcomputeprefix
\end{proposition}

The reader may wonder why Algorithm~\ref{alg:compute-prefix} uses a floating assignment
$\Float{\vy_P}{f_P(\vx)}$, given that it must also perform exact comparisons to obtain the sets $Z$ and $P$ and to decide exactly
 whether $\vy_P \prec f_P(\vy)$ holds in the \textbf{such that} clause of the floating assignment. The reason is that, while we perform such
operations exactly, we do not want to use the {\em result} of exact computations as input for other computations,
as this easily leads to an explosion in the required precision. For instance, the size of the exact result
of~$f_P(\vy)$ may be larger than the size of~$\vy$, while an approximation of smaller size
may already satisfy the \textbf{such that} clause.
In order to emphasize this, we {\em never}
store the result of an exact numerical computation in a variable.

\subsubsection{Computing lower and upper bounds} \label{sub:upper-and-lower}
\mbox{}
Algorithm~~\ref{alg:compute-prefix} uses Kleene iteration
$\vzero, f(\vzero), f(f(\vzero)), \ldots$ to compute a strict pre-fixed point.
One could, in principle, use the same scheme to compute lower bounds of~$\mu_f$,
as this sequence converges to~$\mu_f$ from below by Kleene's theorem. However, convergence of Kleene iteration is generally slow.
It is shown in~\cite{etessamiyannakakis} that for the $1$-dimensional PSP~$f$ with $f(X) = 0.5 X^2 + 0.5$ we have $\mu_f = 1$,
 and the $i$-th Kleene approximant~$\ks{i}$ satisfies $\ks{i} \le 1 - \frac{1}{i}$.
Hence, Kleene iteration may converge only logarithmically, i.e., the number of accurate bits is a logarithmic function of the number of iterations.

In~\cite{etessamiyannakakis} it was suggested to use Newton's method for faster convergence.
In order to see how Newton's method can be used, observe that instead of computing~$\mu_f$, one can equivalently compute the least nonnegative zero
 of \mbox{$f(\vX)-\vX$}.
Given an approximant~$\vx$ of~$\mu_f$, Newton's method first computes $g^{(\vx)}(\vX)$, the first-order linearization of~$f$ at the point~$\vx$:
 \[
  g^{(\vx)}(\vX) = f(\vx) + f'(\vx) (\vX - \vx)
 \]
The next Newton approximant~$\vy$ is obtained by solving $\vX = g^{(\vx)}(\vX)$, i.e.,
 \[
  \vy = \vx + (\Id - f'(\vx))^{-1} (f(\vx) - \vx)\;.
 \]
We write $\Ne_f(\vx) := \vx + (\Id - f'(\vx))^{-1} (f(\vx) - \vx)$, and usually drop the subscript of~$\Ne_f$.
If $\ns{0} \le \mu_f$ is any pre-fixed point of~$f$, for instance $\ns{0} = \vzero$,
 we can define a {\em Newton sequence} $(\ns{i})_i$ by setting $\ns{i+1} = \Ne(\ns{i})$ for $i \ge 0$.
It has been shown in~\cite{etessamiyannakakis,EKL07:stoc,EKL08:stacs} that Newton sequences
 converge at least linearly to~$\mu_f$. 
Moreover, we have $\vzero \le \ns{i} \le f(\ns{i}) \le \mu_f$ for all~$i$.

These facts were shown only for Newton sequences that are computed exactly,
i.e., without rounding errors.
Unfortunately, Newton approximants are hard to compute exactly: Since
each iteration requires to solve a linear equation system whose coefficients
depend on the results of the previous iteration,
the size of the Newton approximants easily explodes.
Therefore, we wish to use inexact arithmetic, but without losing the good properties of Newton's
method (reliable lower bounds, linear convergence).

Algorithm~\ref{gesamtalgo} accomplishes these goals,
and additionally computes post-fixed points~$\ub$ of~$f$, which are upper bounds on~$\mu_f$.%
\begin{algorithm}[ht]
\incmargin{1em}
\linesnumbered
\KwIn{perfectly superlinear PSP $f$, error bound $\epsilon > 0$}
\KwOut{vectors $\lb, \ub$ such that $\lb \le \mu_f \le \ub$ and $\ub - \lb \le \overline{\epsilon}$}
\label{alg:bounds}

  $\mathbf{lb} \leftarrow \texttt{computeStrictPrefix}(f)$\;
  $\mathbf{ub} \leftarrow \overline{1}$\;

  \While{$\mathbf{ub}-\mathbf{lb} \not\le \overline{\epsilon}$}
  {
   $\Float{\vx}{\Ne(\Ne(\lb))}$ \textbf{such that} $f(\lb) + f'(\lb) ( \vx - \lb ) \prec \vx \prec f(\vx) \prec \vone$\; \nllabel{u1}
   $\lb \leftarrow \vx$\;
   $Z \leftarrow \{i \mid 1 \le i \le n, f_i(\ub) = 1\}$\; \nllabel{u2}
   $P \leftarrow \{i \mid 1 \le i \le n, f_i(\ub) < 1\}$\;
   $\vy_Z \leftarrow \overline{1}$\; \nllabel{u2b}
   $\Float{\vy_P}{f_P(f(\ub))}$ \textbf{such that} $f_P(\vy) \prec \vy_P \prec f_P(\ub)$\; \nllabel{u3}
   \ForAll{\upshape superlinear SCCs $S$ of $f$ with $\vy_S = \overline{1}$}
   { \nllabel{u4}
     $\vt \leftarrow \vone - \lb_S$\;
     \If{$f_{SS}'(\vone) \vt \succ \vt$}
     { \nllabel{u6}
       $\displaystyle \Float{\vy_S}{\vone -
         \min \left\{1, \frac{\min_{i\in S}(f_{SS}'(\vone) \vt - \vt)_i}{2 \cdot \max_{i\in S} (f_S(\vtwo))_i} \right\} \cdot \vt}$ 
               \textbf{such that} $f_S(\vy) \prec \vy_S \prec \vone$\; \nllabel{u7}
     }
   }
   $\ub \leftarrow \vy$\; \nllabel{u-end}
 }
\caption{Procedure \texttt{calcBounds}}
\label{gesamtalgo}
\end{algorithm}
Let us describe the algorithm in some detail.
The lower bounds are stored in the variable~$\lb$.
The first value of~$\lb$ is not simply~$\vzero$, but is computed by $\texttt{computeStrictPrefix}(f)$,
in order to guarantee the validity of the following floating assignments.
We use Newton's method for improving the lower bounds because it converges fast (at least linearly) when performed exactly.
In each iteration of the algorithm, {\em two} Newton steps are performed using inexact arithmetic.
The intention is that two inexact Newton steps should improve the lower bound at least as much as one exact Newton step.
While this may sound like a vague hope for small rounding errors,
 it can be rigorously proved thanks to the \textbf{such that} clause of the floating assignment in line~\ref{u1}.
The proof involves two steps.
The first step is to prove that $\Ne(\Ne(\lb))$ is a (strict) post-fixed point of the function $g(\vX) = f(\lb) + f'(\lb) (\vX - \lb)$,
 i.e., $\Ne(\Ne(\lb))$ satisfies the first inequality in the \textbf{such that} clause.
For the second step, recall that $\Ne(\lb)$ is the least fixed point of~$g$.
By Knaster-Tarski's theorem, $\Ne(\lb)$ is actually the least post-fixed point of~$g$.
So, our value~$\vx$, the inexact version of~$\Ne(\Ne(\lb))$, satisfies $\vx \ge \Ne(\lb)$,
 and hence two inexact Newton steps are in fact at least as ``fast'' as one exact Newton step.
Thus, the~$\lb$ converge linearly to~$\mu_f$.

The upper bounds $\ub$ are post-fixed points, i.e., $f(\ub) \le \ub$ is an invariant of the algorithm.
The algorithm computes the sets $Z$ and~$P$ so that inexact arithmetic is only applied to the components~$i$ with $f_i(\ub) < 1$.
In the $P$-components, the function~$f$ is applied to~$\ub$ in order to improve the upper bound.
In fact, $f$ is applied twice in line~\ref{u3}, analogously to applying $\Ne$ twice in line~\ref{u1}.
Here, the \textbf{such that} clause makes sure that the progress towards~$\mu_f$ is at least as fast as the
progress of one exact application of~$f$ would be.
One can show that this leads to linear convergence to~$\mu_f$.

The rest of the algorithm (lines~\ref{u4}-\ref{u7}) deals with the problem that, given a post-fixed~$\ub$, the sequence
 $\ub, f(\ub), f(f(\ub)), \ldots$ does not necessarily converge to~$\mu_f$.
For instance, if $f(X) = 0.75 X^2 + 0.25$, then $\mu_f = 1/3$, but $1 = f(1) = f(f(1)) = \cdots$.
Therefore, the if-statement of Algorithm~\ref{gesamtalgo} allows to improve the upper bound from~$\vone$ to a post-fixed point less than~$\vone$,
 by exploiting the lower bounds~$\lb$.
This is illustrated in Figure~\ref{fig:incons-witness} for a $2$-dimensional scPSP~$f$.
\begin{figure}[ht]
 \begin{tabular}{cc}
 \includegraphics[width=0.48\textwidth]{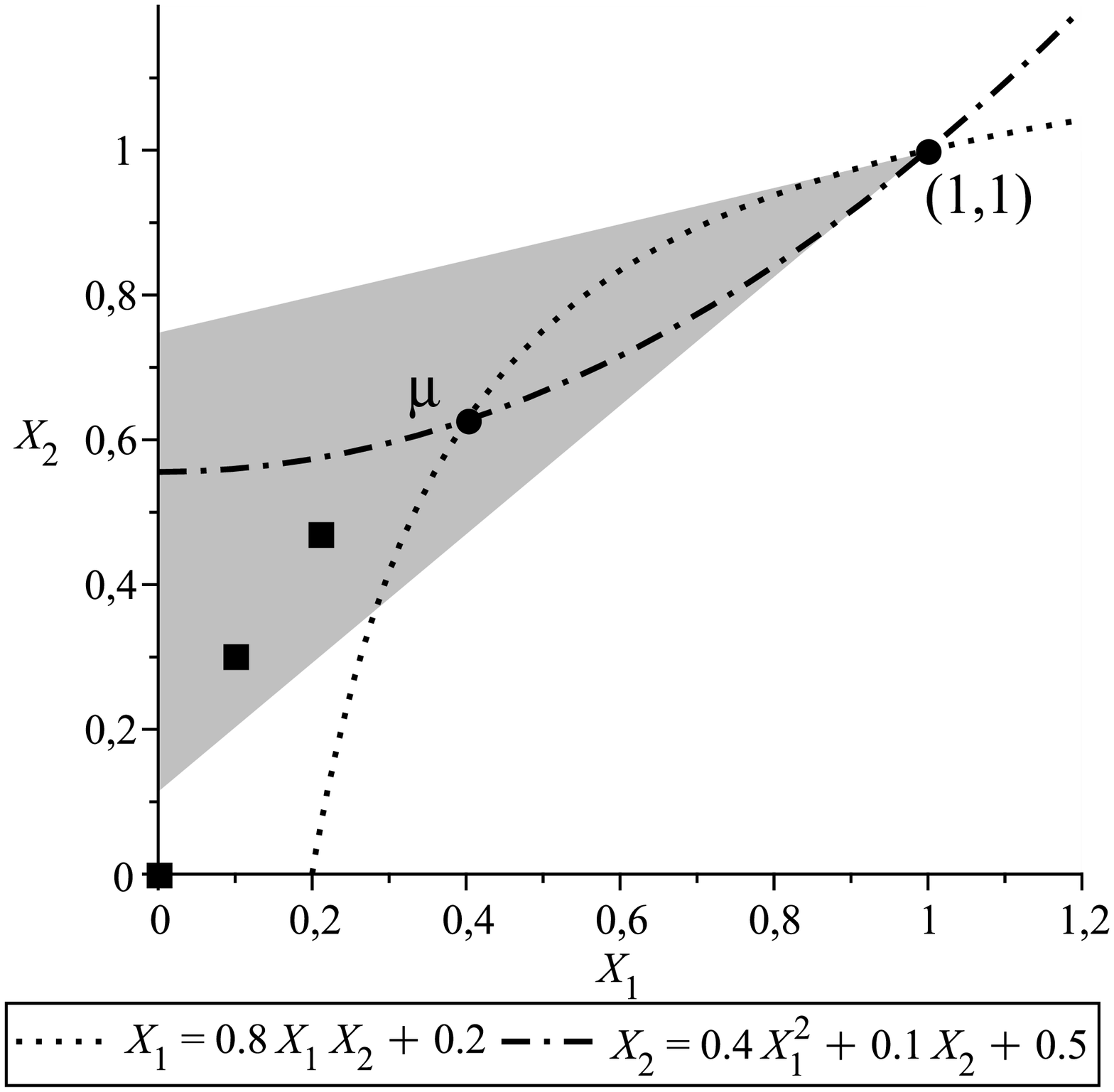} &
 \includegraphics[width=0.48\textwidth]{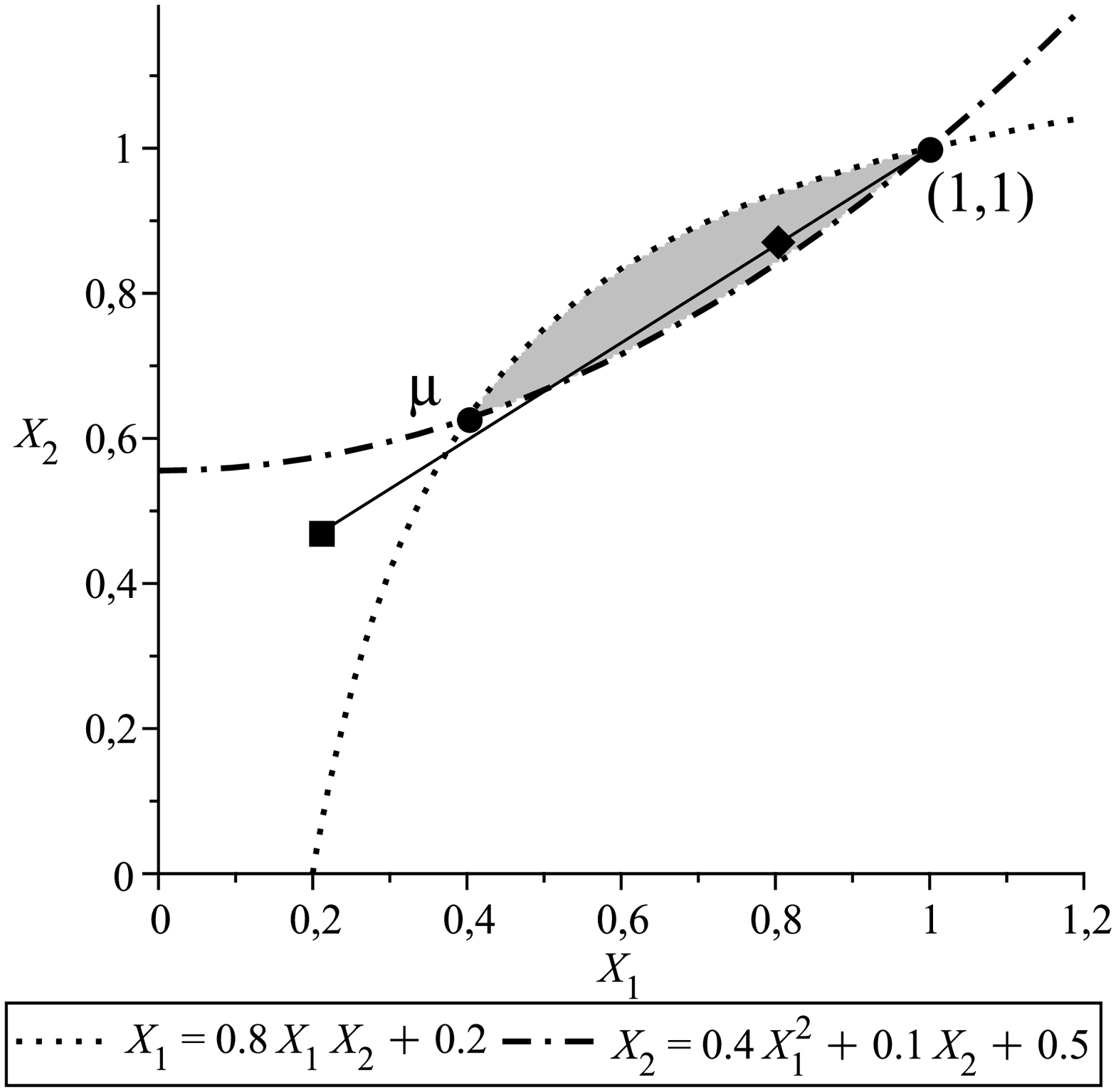} \\
 (a) & (b)
 \end{tabular}
 \caption{Computation of a post-fixed point less than~$\vone$.}
 \label{fig:incons-witness}
\end{figure}
The dotted lines indicate the curve of the points $(X_1, X_2)$ satisfying
 $X_1 = 0.8 X_1 X_2 + 0.2$ and $X_2 = 0.4 X_1^2 + 0.1 X_2 + 0.5$.
Notice that $\mu_f \prec \vone = f(\vone)$.
In Figure~\ref{fig:incons-witness}~(a) the shaded area consists of those points~$\lb$ where $f'(\vone) (\vone - \lb) \succ \vone - \lb$ holds,
 i.e., the condition of line~\ref{u6}.
One can show that $\mu_f$ must lie in the shaded area, so by continuity, any sequence converging to~$\mu_f$, in particular the sequence of lower bounds~$\lb$,
 finally reaches the shaded area.
In Figure~\ref{fig:incons-witness}~(a) this is indicated by the points with the square shape.
Figure~\ref{fig:incons-witness}~(b) shows how to exploit such a point~$\lb$ to compute a post-fixed point~$\ub \prec \vone$
 (post-fixed points are shaded in Figure~\ref{fig:incons-witness}~(b)):
The post-fixed point~$\ub$ (diamond shape) is obtained by starting at~$\vone$
 and moving a little bit along the straight line between $\vone$ and~$\lb$, cf.~line~\ref{u7}.
The sequence $\ub, f(\ub), f(f(\ub)), \ldots$ now converges linearly to~$\mu_f$.

\newcommand{\stmtthmgesamtalgo}{
 Algorithm~\ref{gesamtalgo} terminates and computes vectors $\lb, \ub$ such that $\lb \le \mu_f \le \ub$ and $\ub - \lb \le \overline{\epsilon}$.
 Moreover, the sequences of lower and upper bounds computed by the algorithm both converge linearly to~$\mu_f$.
}
\begin{theorem} \label{thm:gesamt-algo}
 \stmtthmgesamtalgo
\end{theorem}
Notice that Theorem~\ref{thm:gesamt-algo} is about the convergence speed of the
 approximants, not about the time needed to compute them.
To analyse the computation time, one would need stronger requirements
 on how floating assignments are performed.

The lower and upper bounds computed by Algorithm~\ref{gesamtalgo} have a special feature:
 they satisfy $\lb \prec f(\lb)$ and $\ub \ge f(\ub)$.
The following proposition guarantees that such points are in fact lower and upper bounds.

\newcommand{\stmtpropprefixpostfix}{
 Let $f$ be a perfectly superlinear PSP.
 Let $\vzero \le \vx \le \vone$.
 If $\vx \prec f(\vx)$, then $\vx \prec \mu_f$.
 If $\vx \ge f(\vx)$, then $\vx \ge \mu_f$.
}
\begin{proposition} \label{prop:prefix-postfix}
 \stmtpropprefixpostfix
\end{proposition}
\noindent So a user of Algorithm~\ref{gesamtalgo} can immediately verify that the computed bounds are correct.
To summarize, Algorithm~\ref{gesamtalgo} computes provably and even verifiably correct
 lower and upper bounds,
 although exact computation is restricted to detecting numerical problems.
 See Section~\ref{sub:case-neutron} for experiments.

\subsection{Proving consistency using the inexact algorithm}
In Section~\ref{sec:consistency} we presented a simple and efficient algorithm to check the consistency of a PSP.
Algorithm~\ref{gesamtalgo} is aimed at approximating~$\mu_f$, but note that it can also prove the inconsistency of a PSP:
 when the algorithm sets $\ub_i < 1$, we know $(\mu_f)_i < 1$.
This raises the question whether Algorithm~\ref{gesamtalgo} can also be used for proving consistency.
The answer is yes, and the procedure is based on the following proposition.
\newcommand{\stmtpropproveconsinex}{
  Let $f$ be an scPSP.
  Let $\vt \succ \vzero$ be a vector with $f'(\vone) \vt \le \vt$.
  Then $f$ is consistent.
}
\begin{proposition} \label{prop:prove-cons-inex}
 \stmtpropproveconsinex
\end{proposition}
\noindent Proposition~\ref{prop:prove-cons-inex} can be used to identify consistent components.

Use Algorithm~\ref{gesamtalgo} with some (small) $\epsilon$ to compute $\ub$ and $\lb$.
Take any bottom SCC~$S$.
\begin{itemize}
 \item
   If $f'(\vone) (\vone - \lb_S) \le \vone - \lb_S$, mark all variables in~$S$ as consistent and remove the $S$-components from~$f$.
   In the remaining components, replace all variables in~$S$ with~$1$.
 \item
   Otherwise, remove~$S$ and all other variables that depend on~$S$ from~$f$.
\end{itemize}

Repeat with the new bottom SCC until all SCCs are processed.

\noindent There is no guarantee that this method detects all~$i$ with $(\mu_f)_i = 1$.


\subsection{Case study: A neutron branching process} \label{sub:case-neutron}

One of the main applications of the theory of branching processes is
the modelling of cascade creation of particles in physics. We
study a problem described by Harris in~\cite{harris}. Consider a ball of
fissionable radioactive material of radius~$D$. Spontaneous fission
of an atom can liberate a neutron, whose collision with another
atom can produce further neutrons etc. If $D$ is very small, most neutrons
leave the ball without colliding.
If $D$ is very large, then nearly all neutrons eventually collide,
and the probability that the neutron's progeny never dies is large.
A well-known result
shows that, loosely speaking, the population of a process that does not go extinct
grows exponentially over time with large probability. Therefore, the
neutron's progeny never dying out actually means that after a (very)
short time all the material is fissioned, which amounts to a nuclear explosion.
The task is to compute the largest value of~$D$ for which the probability of
extinction of a neutron born at the centre of the ball is still~$1$
(if the probability is $1$ at the centre, then it is $1$ everywhere).
This is often called the critical radius.
Notice that, since the number of atoms that undergo spontaneous fission is large (some hundreds per
second for the critical radius of plutonium), if the probability of extinction lies only slightly below~1, there is already a large probability of a chain reaction. Assume that
a neutron born at distance $\xi$ from the centre leaves
the ball without colliding with probability $l(\xi)$, and collides
with an atom at distance $\eta$ from the centre with
probability density $R(\xi,\eta)$. Let further $f(x) = \sum_{i \geq 0} p_i x^i$,
where $p_i$ is the probability that a collision generates $i$ neutrons.
For a neutron's progeny to go extinct, the neutron must either leave the ball
without colliding, or collide at some distance~$\eta$ from the centre, but in
such a way that the progeny of all generated neutrons goes extinct.
So the extinction probability $Q_D(\xi)$ of a neutron born at distance~$\xi$ from the centre is given by~\cite{harris},~p.~86:
$$Q_D(\xi) = l(\xi) + \int_0^D R(\xi,\eta) f(Q_D(\eta))\; d\eta$$
Harris takes $f(x) = 0.025 + 0.830 x + 0.07x^2 + 0.05 x^3 + 0.025 x^4$,
and gives expressions for both $l(\xi)$ and $R(\xi,\eta)$. By discretizing
the interval $[0,D]$ into $n$ segments and replacing the integral by
a finite sum we obtain a PSP of dimension $n+1$ over the variables
$\{ Q_D(jD/n) \mid 0\leq j \leq n\}$.
Notice that $Q_D(0)$ is the probability that a neutron born in the centre does not cause an explosion.

\paragraph{\em Results}
For our experiments we used three different discretizations $n=20,50,100$.
We applied our consistency algorithm from Section~\ref{sec:consistency} and Maple's Simplex to check inconsistency, i.e., to check whether an explosion occurs.
The results are given in the first 3 rows of Table~\ref{tab:qd0bounds}: Again our algorithm dominates the LP approach, although the polynomials are much denser than in the $\hs{n}$-systems.
\iftechrep{\begin{center}}{}
\begin{table}
\begin{tabular}[c]{l|rrr|rrr|rrr|rrr}
\hspace{1.75cm} $D$                           & \multicolumn{3}{c|}{2}   & \multicolumn{3}{c|}{3} & \multicolumn{3}{c|}{6} & \multicolumn{3}{c}{10} \\
\hspace{1.75cm} $n$                           & 20       & 50     & 100      & 20 & 50 & 100 & 20 & 50 & 100 & 20 & 50 & 100 \\ \hline
inconsistent (yes/no)                         &   n      &   n    &  n       & y & y  & y   & y  & y  & y   & y & y  & y  \\
Cons.~check (Alg. Sec.~\ref{sec:consistency}) &   $<1$   & $<1$   &  2       & $<1$ & $<1$  & 2  & $<1$  & $<1$  & 2  & $<1$ & $<1$  & 2 \\
Cons.~check (exact LP)                        &   $<1$   &  20    &  258     & $<1$ & 22 & 124 & $<1$  & 16 & 168 & $<1$ & 37  & 222\\
Approx.~$Q_D$ ($\epsilon = 10^{-3}$)             &  $<1$   & $<1$   &  4       & 2  &  8   & 32  &  1 & 5  & 21  &  1 & 4 & 17 \\
Approx.~$Q_D$ ($\epsilon = 10^{-4}$)            &  $<1$   & $<1$   &  4       & 2  &  8   & 34  &  2 & 7  & 28 &   1 & 6  & 23
\end{tabular}
\caption{Runtime in seconds of various algorithms on different values of~$D$ and~$n$.}
\label{tab:qd0bounds}
\end{table}
\iftechrep{\end{center}}{}

We also implemented Algorithm~\ref{alg:bounds} using Maple for computing lower and upper bounds on~$Q_D(0)$ with two different values of the error bound~$\epsilon$.
The runtime is given in the last two rows.
By setting the \emph{Digits} variable in Maple we controlled the precision of Maple's software floating-point numbers for the floating assignments.
In all cases starting with the standard value of~10, Algorithm~\ref{alg:bounds} increased {\em Digits} at most twice by~$5$, resulting in a maximal {\em Digits} value of~$20$.
\iftechrep{\Andreas{Version MIT Bild}
The numerical results, plotted in Figure~\ref{img:critical-radius}, fit in well with the approximations given in~\cite{harris}.

\begin{figure}[tbp]
\centering
\includegraphics[width=0.8\textwidth]{critical_radius.eps}
\caption{$Q_D(0)$ for different values of $D$, $n=100$.}
\label{img:critical-radius}
\end{figure}

As a side note we mention that Algorithm~\ref{alg:bounds} computed an upper bound $\prec \vone$, and thus proved inconsistency,
after the first few iterations in all investigated cases, almost as fast as the consistency algorithm from Section~\ref{sec:consistency}.

\paragraph{\em Computing approximations for the critical radius.}
From the data displayed in Figure~\ref{img:critical-radius} one can suspect that the critical radius,
 i.e., the smallest value of~$D$ for which $Q_D(0) = 1$, lies somewhere between 2.7 and~3.}{We mention that Algorithm~\ref{alg:bounds} computed an upper bound $\prec \vone$, and thus proved inconsistency,
after the first few iterations in all investigated cases, almost as fast as the algorithm from Section~\ref{sec:consistency}.
\paragraph{\em Computing approximations for the critical radius.}
After computing $Q_D(0)$ for various values of $D$ one can suspect that the critical radius,
i.e., the smallest value of~$D$ for which $Q_D(0) = 1$, lies somewhere between 2.7 and~3.
}
We combined binary search with the consistency algorithm from Section~\ref{sec:consistency} to determine the critical radius up to an error of~$0.01$.
During the binary search, the algorithm from Section~\ref{sec:consistency} has to analyze PSPs that come closer and closer to the verge of (in)consistency.
For the last (and most expensive) binary search step that decreases the interval to~$0.01$,
 our algorithm took $\mathord{<}1$, $1$, $3$, $8$ seconds for $n = 20, 50, 100, 150$, respectively.
For $n=150$, we found the critical radius to be in the interval $[2.981,2.991]$.
Harris~\cite{harris} estimates $2.9$.

\section{Conclusions} \label{sec:conclusions}
We have presented a new, simple, and efficient algorithm for checking the consistency of PSPs,
which outperforms the previously existing LP-based method.
We have also described the first algorithm that computes reliable lower and upper bounds on~$\mu_f$.
The sequence of bounds converges linearly to~$\mu_f$.
To achieve these properties without sacrificing efficiency, we use a novel combination of exact and inexact (floating-point)
arithmetic. Experiments on PSPs from concrete branching processes confirm the practicality of our approach.
The results raise the question whether our combination of exact and inexact arithmetic could be transferred to other computational problems.

\subsubsection*{Acknowledgments}
We thank several anonymous referees for pointing out inaccuracies and
helping us clarify certain aspects of the paper. The second author was supported by the DFG Graduiertenkolleg 1480 (PUMA).
We also thank Andreas Reuss for proofreading the manuscript.
\bibliographystyle{plain} 
\bibliography{literatur}

\iftechrep{ \newpage \appendix \section{Proofs of Section~\ref{sec:inex}}

\noindent Here is a restatement of Proposition~\ref{prop:normal-form}.

\begin{qproposition}{\ref{prop:normal-form}}
 \propnormalform
\end{qproposition}
\begin{proof}
 In a first step, we add to the equation system $\vX = f(\vX)$ an $(n+1)$-st equation $\tilde{X} = \frac13 \tilde{X}^2 + \frac23$.
 It is easy to see that the least solution of this equation is $\tilde{X} = 1$.
 In order to make $f$ purely superlinear, we take all components~$f_i$ that are not yet superlinear and multiply a monomial of~$f_i$ by~$\tilde{X}$.
 For instance, if $f_i = \frac14 X_j + \frac13 X_k$, then we replace $f_i$ with $\frac14 X_j \tilde{X} + \frac13 X_k$.
 This transformation does not change the least fixed point in the non-$\tilde{X}$-components.
 Call the resulting PSP again~$f$ for simplicity.
 Notice that $f$ is now purely superlinear.

 In a second step we make sure that all superlinear SCCs are purely superlinear.
 For this, we repeatedly apply a certain operation:
 Let $X_j$ be a variable that occurs in a monomial~$m$ of a component~$f_i$, i.e.,
  there is a monomial~$\tilde{m}$ with $m = X_j \cdot \tilde{m}$.
 The operation that replaces the monomial~$m$ in~$f_i$ with $0.5 \cdot m + 0.5 \cdot f_j \cdot \tilde{m}$ is called {\em substituting} (an occurrence of)~$X_j$.
 It is easy to see that applying this operation to a PSP yields a PSP with the same set of fixed points.
 Notice that substituting does not change the dependency relation between the variables.

 In order to make all superlinear SCCs purely superlinear, we apply a sequence of substituting operations.
 Take a superlinear SCC~$S$ which is not purely superlinear and let $g(S)$ denote the PSP obtained by restricting~$f$ to the $S$-components
  and replacing all variables which are not in~$S$ by the constant~$1$.
 Since $S$ is superlinear and not purely superlinear, the PSP~$g(S)$ is, by definition, superlinear and not purely superlinear.
 So there exist variables $X_i, X_j \in S$ such that $X_i$ directly depends on~$X_j$ in~$g(S)$, and $g(S)_i$ is linear, and $g(S)_j$ is superlinear.
 Substitute the corresponding occurrence of~$X_j$ in~$f_i$.
 This makes $g(S)_i$ superlinear.
 By proceeding this way, at most~$n$ substituting operations suffice to make all superlinear SCCs purely superlinear.

 To make $f$ perfectly superlinear, it remains to make each variable directly depend on itself.
 We achieve that by replacing, for all variables~$X$, the polynomial $f_X$ with $0.5 f_X + 0.5X$.
 It is easy to see that $f$ has the same least fixed point, the sum of the coefficients is still at most~$1$ in all components,
  and no new variable dependencies are created by this operation except that every variable now depends directly on itself.
 So, this operation makes $f$ perfectly superlinear.

 Clearly, the bottleneck of this whole procedure consists of the substituting operations.
 Notice that computing the DAG of SCCs can be done in time $O(s)$ with Tarjan's algorithm.
 The size of each single polynomial at the end of the substituting procedure is $O(s)$,
  so the total size of the resulting PSP is $O(n \cdot s)$. 
\end{proof}

\noindent Here is a restatement of Proposition~\ref{prop:alg-compute-prefix}.
\begin{qproposition}{\ref{prop:alg-compute-prefix}}
 \propalgcomputeprefix
\end{qproposition}
\begin{proof}
  We will prove the following invariant of the algorithm:
  \begin{itemize}
   \item[(a)] $\vzero \le \vx \le f(\vx)$;
   \item[(b)] for all components~$j$ with $(f^i(\vzero))_j > 0$ we have $0 < \vx_j < f_j(\vx)$.
  \end{itemize}
  The invariant implies that the loop terminates after at most $n$ iterations, because $f^n(\vzero) \succ \vzero$ holds as $\mu_f \succ \vzero$.
  The invariant also implies that we have $\vzero \prec \vx \prec f(\vx)$ after the loop terminates.

  So it remains to show the invariant.
  Part~(a) clearly holds throughout the loop because for all components~$j$ either $0 = \vx_j$ holds or $0 < \vx_j < f_j(\vx)$
   is guaranteed by the floating assignment.
  Assume inductively that the invariant holds after $i \ge 0$ iterations.
  It suffices to prove that (b) holds after $i+1$ iterations.
  Let $\xs{i}$ denotes the value of~$\vx$ after $i$ iterations.
  Let $(f^{i+1}(\vzero))_j > 0$.
  Then $f_j(\vzero) > 0$ or there is a monomial in~$f_j$ which consists only of variables~$X_k$ with $(f^{i}(\vzero))_k > 0$.
  In the second case we have, by induction hypothesis part~(b), that $\xs{i}_k > 0$ holds for those variables~$X_k$.
  In both cases it follows $f_j(\xs{i}) > 0$.
  Furthermore, we have by induction hypothesis part~(a) that $f_j(\xs{i}) \le f_j(f(\xs{i}))$ where, in fact, the inequality is
   strict because $X_j$ depends on itself (as $f$ is perfectly superlinear).
  We conclude that $0 < f_j(\xs{i}) < f_j(f(\xs{i}))$, and hence the floating assignment guarantees $0 < \xs{i+1}_j < f_j(\xs{i+1})$.
  So the invariant holds after $i+1$ iterations. 
\end{proof}

\subsection{Proof of Proposition~\ref{prop:prefix-postfix}}

For the proof of Proposition~\ref{prop:prefix-postfix} we need some auxiliary lemmas that will also used later on for Theorem~\ref{thm:gesamt-algo}.
We start with the following simple lemma on linear PSPs.
\begin{lemma} \label{lem:lin-unique-fixed-point}
 Let $f$ be a linear PSP.
 Then $\mu_f$ is the unique fixed point of~$f$.
\end{lemma}
\begin{proof}
 Recall that $\mu_f \succ \vzero$.
 Assume that there is a fixed point~$\vx$ of~$f$ different from~$\mu_f$.
 Then, by the linearity of~$f$, all points on the straight line through $\vx$ and $\mu_f$ are fixed points of~$f$.
 So there is a point~$\vy \ge \vzero$ on this straight line with $\vy_i = 0$ for some $i \in \{1,\ldots,n\}$.
 This contradicts the fact that $\mu_f$ is the {\em least} fixed point of~$f$. 
\end{proof}

The following lemma shows how to compute a post-fixed point that satisfies certain properties and is arbitrarily close to~$\vone$.
\begin{lemma} \label{lem:close-to-one}
 Let $f$ be a perfectly superlinear PSP.
 Let $r \in \R$ with $0 < r < 1$.
 Let $\vx = \mu_f + r (\vone - \mu_f)$.
 Then $f(\vx) \le \vx$.
 Furthermore, let $\vp = f^n(\vx)$.
 Then $f(\vp) \le \vp$ and $f_i(\vp) < \vp_i$ holds for all $i \in \{1,\ldots,n\}$ with $(\mu_f)_i < 1$.
\end{lemma}
\begin{proof}
 Letting $\vu, \vv$ be any vectors, we write $f(\vu + \vv) = f(\vu) + f'(\vu) \vv + R(\vu,\vv)$ for the Taylor expansion of~$f$ at~$\vu$.
 Then we have
 \begin{align*}
  \mu_f + f'(\mu_f)(\vone - \mu_f) + R(\mu_f,\vone - \mu_f)
   & = f(\mu_f + (\vone - \mu_f)) = f(\vone) \le \vone = \mu_f + (\vone-\mu_f) \;,
 \end{align*}
  so it follows $f'(\mu_f) \vd + R(\mu_f,\vd) \le \vd$ where $\vd := \vone - \mu_f$.
 Moreover, we have
 \begin{align*}
  f(\vx)
   & = f(\mu_f + r \vd) = f(\mu_f) + f'(\mu_f) r \vd + R(\mu_f,r \vd) \\
   & =   \mu_f + r f'(\mu_f) \vd + R(\mu_f,r \vd) \\
   & \le \mu_f + r f'(\mu_f) \vd + r R(\mu_f,\vd) && \text{(since $R(\mu_f,\cdot)$ is superlinear)} \\
   &  =  \mu_f + r (f'(\mu_f) \vd + R(\mu_f,\vd)) \\
   & \le \mu_f + r \vd && \text{(from above)} \\
   & =   \vx \;,
 \end{align*}
  i.e., $\vx$ is a post-fixed point, hence $\mu_f \le \vx \le \vone$.
 Consider the sequence $\vx \ge f(\vx) \ge f(f(\vx)) \ge \cdots$.
 Since every component depends directly on itself, we have for all components~$i$ that once $(f^{j}(\vx))_i > (f^{j+1}(\vx))_i$ holds for some~$j$,
  we have $(f^{k}(\vx))_i > (f^{k+1}(\vx))_i$ for all $k \ge j$.
 On the other hand, it is easy to see that if $(f^{j}(\vx))_i > (f^{j+1}(\vx))_i$ for some~$j$,
  then $(f^{k}(\vx))_i > (f^{k+1}(\vx))_i$ for some $k \le n$.
 It follows that $f_i(\vp) < \vp_i$ holds for all components~$i$ for which there exists~$j$ with $(f^{j}(\vx))_i > (f^{j+1}(\vx))_i$.
 It remains to show that for all components~$i$ with $(\mu_f)_i < 1$ there exists~$j$ with $(f^{j}(\vx))_i > (f^{j+1}(\vx))_i$.
 Assume for a contradiction that this does not hold.
 Choose a variable~$X_i$ that violates the property (i.e., $(\mu_f)_i < 1$ and $\vx_i = (f^{j}(\vx))_i$ for all~$j$)
  so that all components in lower SCCs satisfy the property.
 Then for all $X_j$ on which $X_i$ depends we have $(\mu_f)_j = 1$, and so $\vx_j = 1$.
 Furthermore, we have $(\mu_f)_S \prec \vx_S \prec \vone$.
 Let $S$ denote the SCC of~$X_i$ and let $g$ be the PSP obtained by restricting~$f$ to the $S$-components and replacing all variables in lower SCCs by the constant~$1$.
 Notice that Lemma~\ref{lem:lin-unique-fixed-point} guarantees that $g$ is not linear, since both $(\mu_f)_S$ and $\vx_S$ are fixed points of~$g$.
 Hence, $g$ is superlinear.
 For any vectors $\vu, \vv$ we write $g(\vu + \vv) = g(\vu) + g'(\vu) \vv + T(\vu,\vv)$ for the Taylor expansion of~$g$ at~$\vu$.
 We have:
 \begin{align*}
       r \vd_S
   & = \vx_S - (\mu_f)_S \\
   & = g(\vx_S) - (\mu_f)_S \\
   & = g((\mu_f)_S + r\vd_S) - (\mu_f)_S \\
   & = (\mu_f)_S + g'((\mu_f)_S) r \vd_S + T((\mu_f)_S,r\vd_S) - (\mu_f)_S \\
   & = g'((\mu_f)_S)r\vd_S + T((\mu_f)_S,r\vd_S)
 \end{align*}
 Moreover, as $\vd_S \succ \vzero$ and $g$ is superlinear, the following inequality is strict in at least one component:
 \begin{align*}
        g(\vone)
  & =   g \left((\mu_f)_S + \frac1r r \vd \right) \\
  & =   g((\mu_f)_S) + g'((\mu_f)_S) \frac1r r \vd_S + T\left((\mu_f)_S,\frac1r r \vd\right) \\
  & \ge (\mu_f)_S + \frac1r g'((\mu_f)_S) r \vd_S + \frac1{r^2} T\left((\mu_f)_S,r \vd\right) \\
  & >   (\mu_f)_S + \frac1r \left( g'((\mu_f)_S) r \vd_S + T\left((\mu_f)_S,r \vd\right) \right) \\
  &  =  (\mu_f)_S + \frac1r r \vd_S && \text{(as computed above)} \\
  &  =  (\mu_f)_S + \vd_S  = \vone
 \end{align*}
 This is the desired contradiction as $g(\vone) \le \vone$ should hold since $g$ is a PSP. 
\end{proof}

The following lemma is used for the proof of Proposition~\ref{prop:prefix-postfix},
 but will also be essential to prove the convergence statement of Theorem~\ref{thm:gesamt-algo}.

\begin{lemma} \label{lem:kleene-converges-linearly}
Let $f$ be a perfectly superlinear PSP.
Let $f(\vp) \le \vp \le \vone$ and $f_i(\vp) < \vp_i$ for all $i\in\{1,\ldots,n\}$ with $(\mu_f)_i < 1$.
Then the sequence $(\ps{i})_{i \in \mathbb{N}}$
defined by
\begin{equation*}
\ps{0} := \mathbf{p} \text{ and } \ps{i+1} := f(\ps{i}) \text{ for } i \geq 0
\end{equation*}
converges linearly to~$\mu_f$.
\end{lemma}
\begin{proof}
If $(\mu_f)_i = 1$ then $(\mu_f)_j = 1$ has to hold for every component $j$ on which $i$ depends.
As $\mu_f$ is the least post-fixed point by Knaster-Tarski's theorem, we have $f^k_i(\mathbf{p}) = 1$ for every $k \in \mathbb{N}$.
Hence we can ignore the 1-components in our convergence proof
 and assume w.l.o.g.\ that $\mu_f \prec \overline{1}$ and with the assumptions $ f(\vp) \prec \vp $.
By the monotonicity of~$f$ and as every variable depends on itself, we get by a simple induction that
 $\ps{i} \succ \ps{i+1} \succ \mu_f$ for all~$i \in \mathbb{N}$.
This already shows that $(\ps{i})$ converges to some limit point.

\newcommand{\Delvu}{\mathbf{\Delta_{u}}}
\newcommand{\Delvp}{\mathbf{\Delta_{p}}}
For every $\mathbf{u} \succ \mu_f$
with $\mathbf{u} \succ f(\mathbf{u})$
we write $\mathbf{u} = \mu_f + \Delvu$ and get:
\begin{align*}
f(\mathbf{u}) - \mu_f & = f(\mu_f+\Delvu) - \mu_f\\
& = f(\mu_f) + f'(\mu_f)\Delvu+R(\mu_f,\Delvu) - \mu_f && \text{(Taylor expansion)}\\
\intertext{Since $R(\mu_f,\Delvu)$ depends at least quadratically on~$\Delvu$,
 one can write $R(\mu_f,\Delvu) = \tilde{R}(\mu_f,\Delvu) \cdot \Delvu$ for a nonnegative matrix $\tilde{R}(\mu_f,\Delvu)$.
 Continuing the above equaliy, we obtain:}
& = (f'(\mu_f) + \tilde{R}(\mu_f,\Delvu))\Delvu \\
& \prec \Delvu && \text{(as $\mathbf{u} \succ f(\mathbf{u})$.)}
\end{align*}
Define $A(\vu) := f'(\mu_f) + \tilde{R}(\mu_f,\Delvu)$, so that we obtain
\begin{equation}
\label{eq:A_absch}
 f(\mathbf{u}) - \mu_f = A(\vu) \cdot \Delvu \prec \Delvu.
\end{equation}
This holds especially for $\mathbf{u} = \mathbf{p}$.
From the $\prec$-inequality in~\eqref{eq:A_absch} follows that there exists $0 < \delta < 1$ such that
\begin{equation}
f(\mathbf{p}) - \mu_f = A(\vp) \Delvp \leq \delta \Delvp.
\label{thm:postkleene_indanf}
\end{equation}
We now show for every $i \ge 0$ that
$\mathbf{p}^{(i)} - \mu_f = \mathbf{\Delta_{p^{(i)}}} \leq \delta^i \Delvp$ by induction over~$i$.
This implies the linear convergence of $(\mathbf{p}^{(i)})_{i \in \mathbb{N}}$.
The base case $i=1$ is proved by~\eqref{thm:postkleene_indanf}.
For $i>1$ note that if $\overline{0} \leq \mathbf{u} \leq \mathbf{u'}$ and $\overline{0} \leq \mathbf{v} \leq \mathbf{v'}$, we have
$A(\mathbf{u})\mathbf{v} \leq A(\mathbf{u'})\mathbf{v'}$, since $A(\mathbf{\mathbf{u}})$ is nonnegative if $\mathbf{u}$ is nonnegative.
\begin{align*}
f^i(\mathbf{p}) - \mu_f & = f(\mathbf{p}^{(i-1)}) - \mu_f\\
& = A(\mathbf{p}^{(i-1)}) ( \mathbf{p}^{(i-1)} - \mu_f) & \text{(by~\eqref{eq:A_absch})} \\
& \leq  A(\vp)  (\delta^{i-1} \Delvp) & \text{(induction hypothesis)} \\
& =  \delta^{i-1} A(\vp) \Delvp \\
& \leq \delta^{i} \Delvp. & \text{~\eqref{thm:postkleene_indanf}}
\end{align*} 
\end{proof}

\noindent Now we can prove Proposition~\ref{prop:prefix-postfix} which we restate here.
\begin{qproposition}{\ref{prop:prefix-postfix}}
 \stmtpropprefixpostfix
\end{qproposition}
\begin{proof}
 By Knaster-Tarski's theorem, $\mu_f$ is the least post-fixed point; the final statement of the proposition follows.
 It remains to show the first statement.
 By choosing the number~$r$ from Lemma~\ref{lem:close-to-one} large enough we can find a post-fixed point~$\vy$ with $\vx \prec \vy \le \vone$.
 By Lemma~\ref{lem:close-to-one} and Lemma~\ref{lem:kleene-converges-linearly} the sequence $\vy, f(\vy), f(f(\vy)), \ldots$ converges to~$\mu_f$.
 On the other hand, by repeatedly applying~$f$ to both sides of the inequality $\vx \prec \vy$ we obtain that
  $\vx \succ f(\vx) \le f^i(\vx) \le f^i(\vy)$ holds for all $i \ge 0$.
 Since $(f^i(\vy))_i$ converges to~$\mu_f$, we have $\vx \prec \mu_f$. 
\end{proof}

\subsection{Proof of Theorem~\ref{thm:gesamt-algo}}

In order to prove Theorem~\ref{thm:gesamt-algo} we can reuse some results of the previous subsection, but we need some additional lemmas.
The following lemma was essentially proved in~\cite{etessamiyannakakis,EKL08:stacs}.

\begin{lemma} \label{lem:newton-star}
 Let $f$ be a perfectly superlinear PSP and let $\vx \prec f(\vx)$.
 Then 
 \mbox{$\Ne(\vx) = \vx + (f'(\vx))^* (f(\vx) - \vx)$},
  where for any square matrix~$A$, we define the matrix star $A^* = \sum_{i=0}^{\infty}{A^i} = \Id + \sum_{i=1}^{\infty}{A^i}$.
\end{lemma}
\begin{proof}
 By Proposition~\ref{prop:prefix-postfix} we have $\vx \prec \mu_f$.
 For such points~$\vx$ it was shown in~\cite{etessamiyannakakis,EKL08:stacs} that $\rho(f'(\vx)) < 1$.
 By standard matrix facts~\cite{matrixbuch}, the matrix star $A^*$ exists if and only if $\rho(A) < 1$.
 Furthermore, if $A^{*}$ exists, it is equal to $(\text{Id} - A)^{-1}$.
 Hence, $(\Id - f'(\vx))^{-1} = f'(\vx)^*$, and the statement follows. 
\end{proof}
So, Lemma~\ref{lem:newton-star} allows to replace the matrix inverse $(\Id - f'(\vx))^{-1}$ with the matrix star $f'(\vx)^*$
 as long as $\vx \prec f(\vx)$ holds, which will be true in this paper whenever we compute~\mbox{$\Ne(\vx)$}.

The following two lemmas are used to show the validity of the floating assignment in line~\ref{u1} of Algorithm~\ref{gesamtalgo}.
\begin{lemma} \label{lem:newton-step-strict}
 Let $f$ be perfectly superlinear.
 Let $\vzero \prec \vx \prec f(\vx) \prec \vone$ and $\vy = \Ne(\vx)$.
 Then $f(\vx) \prec \vy \prec f(\vy) \prec \vone$.
\end{lemma}
\begin{proof}
 By Lemma~\ref{lem:newton-star} we have $\Ne(\vx) = \vx + f'(\vx)^* (f(\vx) - \vx)$.
 Write $\vDelta = f'(\vx)^* (f(\vx) - \vx)$, i.e., $\vy = \vx + \vDelta$.
 As every variable depends directly on itself, we have $f'(\vx) (f(\vx) - \vx) \succ \vzero$.
 Consequently,
  \begin{align*}
   f(\vx) & \prec  \vx + (f(\vx) - \vx) + f'(\vx) (f(\vx) - \vx) \\
          &    =   \vx + \sum_{i=0,1} f'(\vx)^i (f(\vx) - \vx) \\
          &  \le   \vx + \sum_{i=0}^\infty f'(\vx)^i (f(\vx) - \vx) \\
          &    =   \vx + \vDelta \\
          &    =   \vy\;.
  \end{align*}

 Letting $\vu, \vv$ be any vectors, we write $f(\vu + \vv) = f(\vu) + f'(\vu) \vv + R(\vu,\vv)$ for the Taylor expansion of~$f$ at~$\vu$.
 Notice that $R(\vx,\vDelta) \succ \vzero$, because $\vx \succ \vzero$ and $f$ is purely superlinear.
 Hence we have
  \begin{align*}
    \vy & =     \vx + \vDelta \\
        & \prec \vx + \vDelta + R(\vx,\vDelta)                                               \\
        & =     \vx + f'(\vx)^* (f(\vx) - \vx) + R(\vx,\vDelta) \\
        & =     \vx + (f(\vx) - \vx) + f'(\vx) f'(\vx)^* (f(\vx) - \vx) + R(\vx,\vDelta) \\
        & =     f(\vx) + f'(\vx) \vDelta + R(\vx,\vDelta) \\
        & =     f(\vx + \vDelta) \\
        & =     f(\vy)\;.
  \end{align*}
 By~\cite{etessamiyannakakis,EKL07:stoc} we have $\vy \le \mu_f$.
 By the monotonicity of~$f$ it follows that $f(\vy) \le f(\mu_f) = \mu_f$.
 Using the monotonicity of~$f$ once more and the fact that every variable depends directly on itself, we obtain $\vy \prec f(\vy) \prec f(f(\vy)) \le f(\mu_f) = \mu_f$.
 As $\mu_f \le \vone$, it follows $f(\vy) \prec \vone$. 
\end{proof}

The following lemma will also be used to show the validity of the floating assignment in line~\ref{u1} of Algorithm~\ref{gesamtalgo}.
It says that $\Ne(\Ne(\vx))$ is a post-fixed point of the linearization of~$f$ at~$\vx$.
\begin{lemma} \label{lem:doppelter-Newton}
 Let $f$ be perfectly superlinear.
 Let $\vzero \prec \vx \prec f(\vx)$.
 Let $\vz = \Ne(\Ne(\vx))$.
 Then $f(\vx) + f'(\vx)(\vz - \vx) \prec \vz$.
\end{lemma}
\begin{proof}
 Letting $\vu, \vv$ be any vectors, we write $f(\vu + \vv) = f(\vu) + f'(\vu) \vv + R(\vu,\vv)$ for the Taylor expansion of~$f$ at~$\vu$.
 We write $\vy = \Ne(\vx)$ and $\vDelta = f'(\vx)^*(f(\vx) - \vx)$.
 Notice that $\vy = \vx + \vDelta$.
 We have
 \begin{align*}
   \vz & = \vy + f'(\vy)^* (f(\vy) - \vy) \\
       & = \vx + \vDelta + f'(\vy)^* ( f(\vx + \vDelta) - \vx - \vDelta ) \\
       & = \vx + \vDelta + f'(\vy)^* ( f(\vx) + f'(\vx) \vDelta + R(\vx,\vDelta) - \vx - \vDelta ) \\
       & = \vx + \vDelta + f'(\vy)^* ( ( f(\vx) - \vx )+ f'(\vx) f'(\vx)^*(f(\vx) - \vx) - \vDelta + R(\vx,\vDelta) ) \\
       & = \vx + \vDelta + f'(\vy)^* ( \vDelta - \vDelta + R(\vx,\vDelta) ) \\
       & = \vx + \vDelta + f'(\vy)^* R(\vx,\vDelta) \;.\\
 \end{align*}
 It follows
 \begin{align*}
    f(\vx) + f'(\vx)(\vz - \vx) & = f(\vx) + f'(\vx) \left( \vDelta + f'(\vy)^* R(\vx,\vDelta) \right) \\
                                & = \vx + (f(\vx)-\vx) + f'(\vx) f'(\vx)^* (f(\vx)-\vx) + f'(\vx) f'(\vy)^* R(\vx,\vDelta) \\
                                & = \vx + \vDelta + f'(\vx) f'(\vy)^* R(\vx,\vDelta) \\
                                & \le \vx + \vDelta + f'(\vy) f'(\vy)^* R(\vx,\vDelta) \\
                                & \prec \vx + \vDelta + f'(\vy) f'(\vy)^* R(\vx,\vDelta) + R(\vx,\vDelta) \\
                                & = \vx + \vDelta + f'(\vy)^* R(\vx,\vDelta) \\
                                & = \vz \;.
 \end{align*}
 For the $\mathord{\prec}$-inequality in this inequality chain, notice that, since $\vx \prec f(\vx)$, we have $\vDelta \succ \vzero$,
  and since $f$ is purely superlinear, we have $R(\vx,\vDelta) \succ \vzero$. 
\end{proof}

The following lemma states that $\Ne(\vx)$ is the least post-fixed point of the linearization of~$f$ at~$\vx$.
\begin{lemma} \label{lem:newton-less-than-post-fixed}
  Let $f$ be a PSP.
  Let $\vzero \prec \vx \prec f(\vx)$.
  Let $f(\vx) + f'(\vx)(\vz - \vx) \le \vz$.
  Then $\Ne(\vx) \le \vz$.
\end{lemma}
\begin{proof}
 We write $\vy = \Ne(\vx)$ and $\vDelta = f'(\vx)^*(f(\vx) - \vx)$.
 Notice that $\vy = \vx + \vDelta$.
 We have
 \begin{align*}
  \left(\Id - f'(\vx)\right) (\vz - \vy)
   & = \vz - \vx - \vDelta + f'(\vx) (\vx + \vDelta - \vz) \\
   & = \vz - f(\vx) + ( f(\vx) - \vx ) - f'(\vx)^*(f(\vx) - \vx) + f'(\vx) f'(\vx)^*(f(\vx) - \vx) \\
   & \qquad + f'(\vx) (\vx - \vz) \\
   & = \vz - f(\vx) + f'(\vx) (\vx - \vz) \\
   & \ge \vzero \quad \text{(by assumption)}\;.
 \end{align*}
 It follows that
 \begin{align*}
   \vz - \vy & = \left(\Id - f'(\vx)\right)^{-1} \left(\Id - f'(\vx)\right) (\vz - \vy) \\
             & = f'(\vx)^* \left(\Id - f'(\vx)\right) (\vz - \vy) \ge f'(\vx)^* \vzero  = \vzero\;,
 \end{align*}
  i.e., $\Ne(\vx) = \vy \le \vz$. 
\end{proof}

The following lemma shows the validity of the floating assignment in line~\ref{u7} of Algorithm~\ref{gesamtalgo}.

\begin{lemma} \label{lem:compute-postfix}
  Let $f$ be a purely superlinear PSP. 
  Let $\vzero \prec \vt \prec \vone$ such that $f'(1) \vt \succ \vt$.
  Let
   \[
    \vy = \vone - \min \left\{1,\frac{\min_{i\in \{1,\ldots,n\}}(f'(\vone) \vt - \vt)_i}{2 \cdot \max_{i\in \{1,\ldots,n\}} (f(\vtwo))_i} \right\} \cdot \vt\;.
   \]
  Then $f(\vy) \prec \vy \prec \vone$.
\end{lemma}
\begin{proof}
 Letting $\vu, \vv$ be any vectors, we write $f(\vu + \vv) = f(\vu) + f'(\vu) \vv + R(\vu,\vv)$ for the Taylor expansion of~$f$ at~$\vu$.
 Let $\displaystyle r = \min \left\{1,\frac{\min_{i\in \{1,\ldots,n\}}(f'(\vone) \vt - \vt)_i}{2 \cdot \max_{i\in \{1,\ldots,n\}} (f(\vtwo))_i} \right\}$.
 Then we have:
  \begin{align*}
    R(\vone,-r \vt) & \le R(\vone,r \vt) \\
                     & \le r^2 R(\vone,\vt) && \text{(degree of $R(\vone,\cdot)$ at least $2$)} \\
                     & \le r^2 R(\vone,\vone) && \text{($\vt \le \vone$)} \\
                     & \le r^2 f(\vtwo) && \text{($f(\vone+\vone) = f(\vone) + f'(\vone) \vone + R(\vone,\vone)$)} \\
                     & \le r^2 \max_{i\in \{1,\ldots,n\}} (f(\vtwo))_i \cdot \vone \\
                     & \le r \cdot \frac{\min_{i\in \{1,\ldots,n\}}(f'(\vone) \vt - \vt)_i}{2} \cdot \vone && \text{(definition of~$r$)} \\
                     & \le \frac{r}{2} \cdot (f'(\vone) \vt - \vt) \\
                     & \prec r \cdot (f'(\vone) \vt - \vt) && \text{($f'(\vone) \vt \succ \vt$)} \\
  \end{align*}
 Using this inequality we obtain
  \begin{align*}
    f(\vone - r \vt) &  =  f(\vone) + f'(\vone) \cdot (-r \vt) + R(\vone,-r \vt) \\
                     & \le 1 - r f'(\vone) \vt + R(\vone,-r \vt) && \text{($f(\vone) \le \vone$)} \\
                     & \prec 1 - r f'(\vone) \vt + r \cdot (f'(\vone) \vt - \vt) && \text{(see above)} \\
                     &   =   1 - r \vt
  \end{align*}
\end{proof}

\noindent The following lemma states a monotonicity property.
\begin{lemma}
\label{lem:ubfloat}
Let $f$ be perfectly superlinear.
Let $\vy = f(\vx) \le \vx$ and $f_i(\vx) < \vx_i$ for some $i \in \{1,\ldots,n\}$.
Then $f(\vy) \le \vy$ and $f_i(\vy) < \vy_i$.
\end{lemma}
\proof
We have $f(\vy) = f(f(\vx)) \le f(\vx) = f(\vy)$ by the monotonicity of~$f$.
Moreover, since each component depends on itself, the strict inequality $f_i(\vx) < \vx_i$ implies the strict inequality $f_i(f(\vx)) < f_i(f(\vx))$.
\qed

The following lemma will be used in the proof of Theorem~\ref{thm:gesamt-algo} to show that $\ub_i < 1$ eventually holds in the components~$i$ with $(\mu_f)_i < 1$.
\begin{lemma}
\label{thm:taylor}
Let $f$ be a purely superlinear PSP and $\mathbf{x} \geq \overline{0}, \mathbf{u} \succ \overline{0}$. Then
\begin{equation*}
 f'(\mathbf{x}+\mathbf{u})\mathbf{u} \succ f(\mathbf{x}+\mathbf{u}) - f(\mathbf{x}).
\end{equation*}
\end{lemma}
\proof
It suffices to show $f_i(\mathbf{x})  - f_i(\mathbf{x}+\mathbf{u}) + f'_i(\mathbf{x}+\mathbf{u})\mathbf{u} > 0$ for every component $i$
of $f$.
We can write $f_i(\mathbf{x}+\mathbf{u})$ as $f_i(\mathbf{x}) + \int_{0}^{1}f'_i(\mathbf{x}+s\mathbf{u}) \mathbf{u}\,ds$. Hence
\begin{align*}
& \\
& f_i(\mathbf{x})  - f_i(\mathbf{x}+\mathbf{u}) + f'_i(\mathbf{x}+\mathbf{u})\mathbf{u} \\
  & =
f_i(\mathbf{x}) - f_i(\mathbf{x}) - \int_{0}^{1}{f'_i(\mathbf{x}+s\mathbf{u})\mathbf{u}\,ds}+f'_i(\mathbf{x}+\mathbf{u})\mathbf{u}\\
& =  - \int_{0}^{1}{f'_i(\mathbf{x}+s\mathbf{u})\mathbf{u}\,ds} + f'_i(\mathbf{x}+\mathbf{u})\mathbf{u} \\
& =  - \int_{0}^{1/2}{f'_i(\mathbf{x}+s\mathbf{u})\mathbf{u}\,ds} - \int_{1/2}^{1}{f'_i(\mathbf{x}+s\mathbf{u})\mathbf{u}\,ds}
  + f'_i(\mathbf{x}+\mathbf{u})\mathbf{u} \\
\end{align*}

For $0\leq s \leq 1$, we have $f'_i(\mathbf{x}+s\mathbf{u})\mathbf{u} \le f'_i(\mathbf{x}+\mathbf{u}) \mathbf{u}$,
 and for $0 \le s \le 1/2$, the inequality is strict, because $\vu \succ \vzero$ and $f$ is purely superlinear.
Hence
\begin{align*}
 & - \int_{0}^{1/2}{f'_i(\mathbf{x}+s\mathbf{u})\mathbf{u}\,ds}
- \int_{1/2}^{1}{f'_i(\mathbf{x}+s\mathbf{u})\mathbf{u}\,ds} +f'_i(\mathbf{x}+\mathbf{u})\mathbf{u}\\
& >  - \int_{0}^{\frac{1}{2}}{f'_i(\mathbf{x}+\mathbf{u})\mathbf{u}\,ds}
- \int_{\frac{1}{2}}^{1}{f'_i(\mathbf{x}+\mathbf{u})\mathbf{u}\,ds} + f'_i(\mathbf{x}+\mathbf{u})\mathbf{u}\\
& = - f'_i(\mathbf{x}+\mathbf{u})\mathbf{u} + f'_i(\mathbf{x}+\mathbf{u})\mathbf{u} = 0.
\end{align*}
\qed


\noindent Now we can prove Theorem~\ref{thm:gesamt-algo} which is restated here:
\begin{qtheorem}{\ref{thm:gesamt-algo}}
 \stmtthmgesamtalgo
\end{qtheorem}
\begin{proof}
 The validity of the floating assignment in line~\ref{u1} follows from Lemma~\ref{lem:newton-step-strict} and Lemma~\ref{lem:doppelter-Newton}.
 Next we show the convergence of the lower bounds.
 Let $(\lbs{k})_k$ be the sequence of the lower bounds $\lb$ in the algorithm, where $\lbs{0}$ is the result of $\texttt{computeStrictPrefix}(f)$.
 Moreover, define an ``exact'' Newton sequence $\ns{0} = \lbs{0}$ and $\ns{k+1} = \Ne ( \ns{k} )$.
 We prove by induction that $\ns{k} \le \lbs{k}$.
 The induction base ($k=0$) is trivial.
 Let $k \ge 0$.
 Notice that the floating assignment in line~\ref{u1} guarantees $f(\lbs{k}) + f'(\lbs{k}) \left( \lbs{k+1} - \lbs{k} \right) \le \lbs{k+1}$.
 Therefore, Lemma~\ref{lem:newton-less-than-post-fixed} assures $\Ne(\lbs{k}) \le \lbs{k+1}$.
 Hence we have
 \begin{align*}
  \ns{k+1} &  =  \Ne(\ns{k}) \\
           & \le \Ne(\lbs{k}) && \text{(induction hypothesis, monotonicity of~$\Ne$ as shown in~\cite{EKL08:stacs})} \\
           & \le \lbs{k+1}    && \text{(as argued above)}\;.
 \end{align*}
 So we have $\ns{k} \le \lbs{k}$ for all~$k$.
 By the floating assignment in line~\ref{u1}, we have $\lbs{k} \prec f(\lbs{k})$, so $\lbs{k} \prec \mu_f$ by Proposition~\ref{prop:prefix-postfix}.
 As $(\ns{k})_k$ converges to~$\mu_f$, the sequence $(\lbs{k})_k$ converges to~$\mu_f$ as well.
 In addition, it was shown in~\cite{EKL07:stoc,EKL08:stacs} that $(\ns{k})_k$ converges linearly to~$\mu_f$.
 As $\ns{k} \le \lbs{k}$, the same holds for $(\lbs{k})_k$.

 Now we turn the upper bounds~$\ub$.
 We prove the following invariants of the algorithm:
  \begin{itemize}
   \item[(a)] $f(\ub) \le \ub \le \vone$;
   \item[(b)] for all components~$j$ with $\ub_i < 1$, we have $f_i(\ub) < \ub_i$.
  \end{itemize}
 Clearly, this holds at the beginning (when $\ub = \vone$).
 The invariants are clearly preserved by the assignment in line~\ref{u2b}.
 Repeated application of Lemma~\ref{lem:ubfloat} shows that the floating assignment in line~\ref{u3} is valid and that the invariants are preserved.
 Lemma~\ref{lem:compute-postfix} implies that the floating assignment in line~\ref{u7} are valid and preserve the invariants.
 Hence, the invariants hold.

 Next we prove that for any component~$i$ with $(\mu_f)_i < 1$, we eventually have $\ub_i<1$.
 Let us assume for the sake of a contradiction that there exists a component~$i$ with the property~$P(i)$,
  where $P(i)$ means that $(\mu_f)_i < 1$ and $\ub_i=1$ holds during the entire execution of the algorithm.
 Choose $i$ ``minimal'' in the sense that for all variables~$X_j$ on which $X_i$ depends we have that either $X_i$ and $X_j$ are in the same SCC,
  or $P(j)$ does not hold.
 Let $S$ be the SCC of~$X_i$ and let $X_j$ be any variable from $\text{Var} \setminus S$ on which $X_i$ depends.
 Since $P(i)$ holds, we must have $\ub_j = 1$ during the entire execution of the algorithm,
  because if $\ub_j < 1$ were true at some point, it would take at most $n$ iterations before $\ub_i < 1$.
 As $P(j)$ cannot hold by the minimality of~$i$, we have $(\mu_f)_j = 1$.
 Therefore, letting $g$ denote the PSP obtained by restricting~$f$ to the $S$-components and replacing all variables from other SCCs by the constant~$1$,
  we have $\mu_g = (\mu_f)_S \prec \vone$.
 Since $\ub_S = \vone$ holds during the execution of the algorithm, we have $g(\vone) = \vone$, i.e., $\vone$ is a fixed point of~$g$.
 Therefore, by Lemma~\ref{lem:lin-unique-fixed-point}, $g$ cannot be linear, as $\mu_g \prec \vone$.
 Since $f$ is perfectly superlinear, $g$ must then be purely superlinear.
 Application of Lemma~\ref{thm:taylor} (with $\vx := \mu_g$ and $\vu := \vone - \mu_g$) yields
 \begin{align*}
   g'(\vone) (\vone - \mu_g) \succ g(\vone) - g(\mu_g) = \vone - \mu_g\;.
 \end{align*}
 Since the sequence of $\lb_S$ computed during the execution of the algorithm converges to~$\mu_g$,
  the continuity of $g'(\vone)$ implies that eventually $g'(\vone)(\vone - \lb_S) \succ \vone - \lb_S$ holds.
 But this means that the condition of line~\ref{u6} is satisfied and, thus, the following assignment causes $\ub_S \prec \vone$,
  contradicting our assumption that $P(i)$ holds.
 So we have shown that for any component~$i$ with $(\mu_f)_i < 1$, we eventually have $\ub_i<1$.

 Denote by $(\ubs{k})_k$ the sequence of upper bounds~$\ub$ computed by the algorithm.
 It remains to show that this sequences converges linearly to~$\mu_f$.
 We have shown above that there exists $k_0$ such that for all $k \ge k_0$ we have that
  $\ubs{k+1}_i \le \ubs{k}_i < 1$ holds for all components~$i$ with $(\mu_f)_i < 1$.
 Choose a real number~$r$ with $0 < r <1$ such that for the point $\vp := \mu_f + r (\vone - \mu_f)$ we have $\ubs{k_0} \le \vp \le \vone$ and
  the following is true for all components~$i$: either $(\ubs{k_0})_i = 1$ or $\ubs{k_0}_i < \vp_i < 1$.
 Define the sequence $(\ps{k})_{k \ge k_0}$ by setting $\ps{k_0} := \vp$ and $\ps{k+1} := f(\ps{k})$ for all $k \ge k_0$.
 By Lemma~\ref{lem:close-to-one} and Lemma~\ref{lem:kleene-converges-linearly}, this sequence converges linearly to~$\mu_f$.
 To prove that the same holds for $(\ubs{k})$, it suffices to show that $\ubs{k} \le \ps{k}$ holds for all $k \ge k_0$.
 We proceed by induction on~$k$.
 The induction base ($k=k_0$) holds by definition of~$\ps{k_0}$.
 Let $k \ge k_0$.
 Then we have:
 \begin{align*}
   \ubs{k+1}
    & \le f(\ubs{k}) && \text{(\textbf{such that} clause of line~\ref{u3})} \\
    & \le f(\ps{k})  && \text{(induction hypothesis)} \\
    &  =  \ps{k+1}   && \text{(definition of~$\ps{k+1}$)}
 \end{align*} 
 This completes the proof.
\end{proof}

\subsection{Proof of Proposition~\ref{prop:prove-cons-inex}}
\noindent Here is a restatement of Proposition~\ref{prop:prove-cons-inex}.
\begin{qproposition}{\ref{prop:prove-cons-inex}}
 \stmtpropproveconsinex
\end{qproposition}
\begin{proof}
 Lemma~\ref{lemma:pf1}~(4) implies $\rho(f'(\vone)) \le 1$.
 Hence, $f$ is consistent by Proposition~\ref{prop:spektralradius}.
 \stefan{ Ausfuehrlicher waere es so: fuer alle $\beta > 1$ gilt, dass $\rho(f'(\vone)) < \beta$. Es folgt $\rho(f'(\vone)) \le 1$.} 
\end{proof}
 }{}

\end{document}